\documentclass[journal]{IEEEtran}
\IEEEoverridecommandlockouts
\ifCLASSINFOpdf
\else
\fi

\usepackage{tabularx}
\usepackage{threeparttable}
\usepackage{mathrsfs}
\usepackage{color}
\usepackage{amsfonts,amssymb}
\usepackage{amsfonts}
\usepackage{subfigure}
\usepackage{booktabs}

\usepackage{amsmath}
\usepackage{enumerate}
\usepackage{algorithm}
\usepackage{algorithmic}
\usepackage{bm}
\usepackage{makecell}
\usepackage{multirow}
\usepackage{mathtools}
\usepackage{bbm}
\usepackage{dsfont}
\usepackage{pifont}
\usepackage{amsthm}
\usepackage{cite}
\usepackage{balance}
\usepackage{graphicx}

\newtheorem{remark}{Remark}
\newtheorem{theorem}{Theorem}
\newtheorem{lemma}{Lemma}
\newtheorem{corollary}{Corollary}
\newtheorem{assumption}{Assumption}

\begin{document}

\title{
Topology Inference for Network Systems: Causality Perspective and Non-asymptotic Performance}
\author{Yushan Li, Jianping He, Cailian Chen and Xinping Guan
\thanks{The authors are with Dept. of Automation, Shanghai Jiao Tong University, Key Laboratory of System Control and Information Processing, Ministry of Education of China, and Shanghai Engineering Research Center of Intelligent Control and Management, Shanghai, Chin. E-mail address: \{yushan\_li, jphe, cailianchen, xpguan\}@sjtu.edu.cn. 
Preliminary result of this paper was presented at the 60th IEEE Conference on Decision and Control, 2021 \cite{lys-cdc2021}. 
}%
}

\maketitle

\begin{abstract}
Topology inference for network systems (NSs) plays a crucial role in many areas. 
This paper advocates a causality-based method based on noisy observations from a single trajectory of a NS, which is represented by the state-space model with general directed topology. 
Specifically, we first prove its close relationships with the ideal Granger estimator for multiple trajectories and the traditional ordinary least squares (OLS) estimator for a single trajectory. 
Along with this line, we analyze the non-asymptotic inference performance of the proposed method by taking the OLS estimator as a reference, covering both asymptotically and marginally stable systems. 
The derived convergence rates and accuracy results suggest the proposed method has better performance in addressing potentially correlated observations and achieves zero inference error asymptotically.  
Besides, an online/recursive version of our method is established for efficient computation or time-varying cases. 
Extensions on NSs with nonlinear dynamics are also discussed. 
Comprehensive tests corroborate the theoretical findings and comparisons with other algorithms highlight the superiority of the proposed method. 
\end{abstract}


\IEEEpeerreviewmaketitle

\section{Introduction}

Network systems (NSs) are characterized by the locality of information exchange between individual nodes (described by a topology) \cite{olfati2007consensus}, and the cooperative capability to solve a common task \cite{nokleby2018stochastic}. 
Inferring the interaction topology structure from observations over the system emerges in various applications in last decades, including social networks \cite{ahmed2009recovering}, brain connectivity patterns \cite{monti2014estimating} and multi-robot formation \cite{liu2019dynamic}, to name a few. 
As topology inference helps better understand the systems and implement coordinated tasks, it brings significant benefits for numerous applications of NSs. 
For instance, tracing the information flow over a social network \cite{mahdizadehaghdam2016information}, group testing and identification of defective items \cite{cheraghchi2012graph}, or anomaly detection in communications networks \cite{mardani2013dynamic}.

Mathematically, topology inference can be regarded as a typical inverse modeling problem. 
In the literature, a large body of research has been developed to tackle the problem due to their massive employment \cite{giannakis2018topology}. 
For example, plenty of researchers have considered using graphical models to describe the relationships between different variables, and utilize graph signal processing (GSP) techniques to infer the underlying undirected topology, e.g., see \cite{segarra2017network,egilmez2017graph,zhu2020network}. 
The main idea is to find the most suitable eigenvalues and eigenvectors from the sample correlation matrix and reconstruct the topology. 
Considering the node causality (directionality), \cite{granger1969investigating,brovelli2004beta,santos2019local} utilize Granger estimator to capture the casual relationships between agents. 
Other typical alternatives like structural equation model \cite{baingana2017tracking,ioannidis2019semi} and vector autoregressive analysis \cite{geiger2015causal,zaman2021online} are also developed to deal with the directed topology cases. 
\cite{weerts2018identifiability,hendrickx2019identifiability,vanwaarde2021topology,cheng2022allocation} focus on inferring the topology between different system modules with dynamics (characterized by rational transfer functions), which is a highly abstract representation compared with the common state-space model.

\textit{\textbf{Motivations:}} 
Despite the prominent contributions of the pioneering works, there still remain some notable issues when we focus on the topology inference of NSs. 
First, the GSP-based methods mainly focus on revealing the correlation between the nodes to explain the data regularity (e.g., estimating the inverse covariance matrix \cite{friedman2008sparse}), which is usually characterized by an undirected graph. 
The formulation of the graphical models has less taken the system evolution with time into account, and the observations for inference are generally assumed to be independently sampled from certain random distributions (see \cite{mateos2019connecting,dong2019learning} for a review). 
Therefore, the majority of GSP methods cannot interpret the generally directed topology of NSs that characterize the time causality between nodes. 
Second, most existing directed topology inference methods potentially rely on some prior assumptions about the system model and stability. 
For instance, the Granger estimator is based on the observations from \textit{multiple trajectories}\footnote{A single trajectory refers to the collected sequential observations (outputs) of the system, by starting the system and letting it evolve with time. 
Correspondingly, multiple trajectories are obtained by restarting and running the system from the same initial state multiple times.}, not appropriate for \textit{single trajectory} cases. 
The structural equation model captures the contemporaneous causal dependencies among nodes, without considering the general time-lagged influences. 
Third, the former two issues are promising to be effectively handled by vector autoregressive methods. 
However, related works usually neglect the influence brought by observation noises. 
They are mostly devoted to the effective algorithm designs, and lack the non-asymptotic analysis of the inference performance about the observation number.

Motivated by the above issues, this paper focuses on the directed topology inference of NSs in state-space representation, where the observations are corrupted by noises. 
Specifically, we aim to reveal the relationships between the basic inference principles using observations from multiple and single trajectory of NSs, respectively. 
Meanwhile, we seek to derive the non-asymptotic convergence rate and accuracy of the inference methods about the observation number. 
The challenge of our work is two-fold. 
On the one hand, only noisy observations over the system evolution are available, incurring latent correlation on each consecutive pair in causality modeling. 
On the other hand, the observations contain noise accumulation that is determined by different system stability, making it hard to directly characterize the inference performance.

\textit{\textbf{Contributions:}} 
Preliminary results about the relationships of different inference methods in asymptotically stable cases have appeared in \cite{lys-cdc2021}. 
This paper extends the analysis by i) investigating the mutual relationships covering different system stability, ii) characterizing the non-asymptotic performance of the proposed estimator, and ii) providing extensions to more complicated topology cases. 
The main contributions are summarized as follows.

\begin{itemize}
\item This work contributes to the existing body of research by revealing relationships between the inference principles using single and multiple trajectories. 
Accordingly, we propose a causality-based method to effectively infer the topology from highly correlated and noisy observations in a single trajectory, applying to both asymptotically and marginal stable NSs. 

\item Towards the Granger and proposed estimator, we prove their equivalence conditions in asymptotically stable system cases, and clarify their deviation in sample matrices in marginally stable cases. 
Then, the convergence rate of the proposed estimator in both cases is derived, which eliminates the bounded inference error by OLS estimator and achieves zero error in the asymptotic sense.

\item We prove that the proposed estimator is essentially a de-regularization version of the OLS estimator. 
Following this implication, we provide the online/recursive form of the proposed estimator, which can be applied to time-varying topology cases. 
Finally, extensive simulations verify the theoretical findings, and comprehensive comparisons with state-of-the-art algorithms corroborate the superior performance. 

\end{itemize}


\textit{\textbf{Organizations:}} 
The remainder of this paper is organized as follows. 
Section \ref{section-related work} presents related literatures. 
Section \ref{preliminary} gives basic preliminaries and describes the problem of interest. 
The inference methods along with their relationships are presented in Section \ref{main-results}, 
The convergence rate and accuracy of our method are analyzed in Section \ref{section:Non-asymptotic}. 
Section \ref{section:extension} discusses some extensions to more complicated cases. 
Simulation results are shown in Section \ref{simulation},
followed by the concluding remarks and further research issues in Section \ref{conclusion}.

\section{Related Work}\label{section-related work}

\textit{Static topology inference in linear NSs}. 
The static topology in linear NSs is the most investigated type.   
In \cite{etesami2017measuring}, the authors consider the casual dynamics model and focus on learning the causal relationships by means of functional dependencies. 
Optimization algorithms are designed in \cite{dong2016learning,pasdeloup2018characterization} to infer the graph Laplacian matrix of the network from the nodal observations, by considering the stationary signals are smoothly evolving. 
\cite{hayden2016network} investigates the identifiability conditions for unknown dynamical networks from output second-order statistics, where the network is driven by stochastic inputs. 
In relation to the inference of networks from consensus dynamics, topology is reconstructed by measuring the power spectral density of the network response to input noises \cite{shahrampour2014topology}. 
Aiming at the adaptive diffusion process of the network, the correlation methods are proposed to achieve the progressive approximation over partially observed networks \cite{matta2018consistent,santos2019local}. 
For large-scale networks, many works are developed to reconstruct a sparse topology from limited observations via compressed sensing \cite{timme2007revealing,bolstad2011causal,hayden2016sparse,wai2019joint}, especially when the node number is much more than the available observations. 
The problem is usually transformed into a $L_1$ norm optimization problem (also known Lasso problem).

\textit{Methods on time-varying and nonlinear cases.} 
In many applications, the observations entail a time-varying graph and static graph inference methods will fail to capture the dynamic characteristic. 
Since multiple time intervals are involved, the time-varying topology inference is commonly transformed into a sequential optimization problem with multiple topology variables and the switching point detection, e.g., see \cite{peel2015detecting,baingana2016tracking,hallac2017network}. 
Besides, for NSs with nonlinear system dynamics, kernel-based methods are widely investigated \cite{shen2017kernel,romero2017kernelbased,wang2018inferring}. 
The key idea is to select appropriate kernel basis functions to approximate the nonlinear dynamics, and thus the selection of kernels critically affects the performance. 
In general, most of the above works focus on specific algorithm design, and do not consider the observations can also be corrupted by independent noises. 
In addition, it is less noticed to investigate the non-asymptotic performance of the inference methods in terms of the convergence rates and accuracy.

\textit{Differences with system identification.} 
It is worth noting some commonalities and fundamental differences between our problem and traditional system identification. 
Intuitively, the non-asymptotic analysis manner on the inference deviation can be likewise, where they both exhibit as the matrix computation and norm scaling. 
Nevertheless, the system identification usually aims to identify the system's Markov parameters from known input-output pairs \cite{ljung1998system}. 
In this paper, we only have access to noise-corrupted outputs (observations) and need to reconstruct the topology matrix from the consecutive observation pairs, which are undesirably correlated. 
Besides, various system stability will incur distinct non-asymptotic inference performance. 
Most related works only consider one type stability in single trajectory (e.g., asymptotically stable system in \cite{oymak2019non,kowshik2021streaming}) or broad stability but requiring multiple trajectories (e.g., see \cite{sun2020finite,zheng2021nonasymptotica}). 
The analysis for the proposed estimator covers both asymptotic and marginal stability using only a single trajectory, and needs devoted efforts to tackle the troublesome noise accumulation effects.

\section{Preliminaries and Problem Formulation}\label{preliminary}
\subsection{Graph Basics and Notations}
Let $\mathcal{G}=(\mathcal{V},\mathcal{E})$ be a directed graph that models the network system, where $\mathcal{V}=\{1, \cdots, n\}$ is the finite set of nodes and $\mathcal{E}\subseteq \mathcal{V}\times \mathcal{V}$ is the set of interaction edges.
An edge $(i,j)\in \mathcal{E}$ indicates that $i$ will use information from $j$.
The adjacency matrix $A=[a_{ij}]_{n \times n}$ of $\mathcal{G}$ is defined such that ${a}_{ij}\!>\!0$ if $(i,j)$ exists, and ${a}_{ij}\!=\!0$ otherwise.
Denote ${\mathcal{N}_i}=\{j\in \mathcal{V}:a_{ij}>0\}$ as the in-neighbor set of $i$, and $d_i=\left| {\mathcal{N}_i} \right |$ as its in-degree.

Throughout this paper, the set variable, vector, and matrix are expressed in Euclid, lowercase, and uppercase fonts, respectively. 
Let $\bm{0}$ ($\bm{1}$) be all-zero (all-one) matrix in compatible dimensions. 
Denote by $\rho_{\min}(\cdot)$ and $\rho_{\max}(\cdot)$ the smallest and largest singular values of a matrix, and $\lambda_{\min}(\cdot)$ and $\lambda_{\max}(\cdot)$ represent the smallest and largest eigenvalues of a square matrix. 
For square matrices $M_a$ and $M_b$ in the same dimensions, ${M_a}\!\succeq\!{M_b}$ (${M_a}\!\preceq\!{M_b}$) means $({M_a}-{M_b})$ is positive-semidefinite (negative-semidefinite). 
Unless otherwise noted, $\| \cdot \|$ and $\| \cdot \|_F$ represent the spectral and Frobenius norm of a matrix, respectively. 
For two real-valued functions $f_1$ and $f_2$, $f_1(x)=\bm{O}(f_2(x))$ as $x\to x_0$ means $\mathop {\lim }\nolimits_{x \to x_0 } |f_1(x)/f_2(x)|<\infty$. 


\subsection{System Model}
Consider the following network system model
\begin{equation}\label{eq:global_model}
\begin{aligned}
x_{t}&=Wx_{t-1}+\theta_{t-1}, \\
y_{t}&=x_{t}+\upsilon_{t},
\end{aligned}
\end{equation}
where $x_t$ and $y_t$ represents the system state and corresponding observation at time $t$ ($t=1,2,\cdots, T$), $W\in \mathbb{R}^{n \times n}$ is the unknown topology matrix related to the adjacency matrix $A$, 
and $\theta_{t}$ and $\upsilon$ represent the process and observation noises, satisfying the following Gauss-Markov assumption. 
\begin{assumption}\label{assu:noise}
$\theta_{t}$ and $\upsilon_{t}$ are i.i.d. Gaussian noises, subject to ${N}(0,\sigma^2_{\theta} I)$ and ${N}(0,\sigma^2_{\upsilon} I)$, respectively, and $\sigma^2_{\theta} \ge \sigma^2_{\upsilon} $. 
They are also independent of $\{x_{t}\}_{t=0}^{t=T}$ and $\{y_{t}\}_{t=0}^{t=T}$. 
\end{assumption}
Here $\sigma^2_{\theta} \ge \sigma^2_{\upsilon} $ is considered to avoid that the observation noises may cover up system states when the initial state values are small. 
Next, we present asymptotically stable matrix $\mathcal{S}_a$ and the (strict) marginally stable matrix $\mathcal{S}_m$ as follows:
\begin{equation}
\begin{aligned}
\mathcal{S}_a=&\{Z\in \mathbb{R}^{n \times n}, \rho_{\max}(Z)<1\}, \\
\mathcal{S}_m=&\{Z\in \mathbb{R}^{n \times n}, \rho_{\max}(Z)=1~\text{~and the geometric} \\
&\text{multiplicity of eigenvalue 1 equals to one} \} .
\end{aligned}
\end{equation}
In terms of the setup of $W$, some useful and popular choices are the Laplacian and the Metropolis rules, which are defined as follows \cite{sayed2014adaptation}.
For $i \neq j$,
\begin{align}\label{eq:topo-rule}
w_{ij}=\left \{
\begin{aligned}
&\gamma a_{ij}/ \max\{d_i, i\in\mathcal{V}\},&&\text{by Laplacian rule} ,\\
&a_{ij}/ \max \{d_{i}, d_{j}\},&&\text{by Metropolis rule},
\end{aligned}
\right.
\end{align}
where the auxiliary parameter $\gamma$ satisfies $0<\gamma\le 1$.
For both rules, the self-weights are given by
\begin{align}\label{eq:topo-rule2}
w_{ii}= 1-\sum\limits_{j \neq i} w_{ij}.
\end{align}
Note that if $W$ is specified by either one of the two rules, then $W\in\mathcal{S}_m$.
A typical matrix in $\mathcal{S}_a$ can be directly obtained via multiplying (\ref{eq:topo-rule}) and (\ref{eq:topo-rule2}) by a factor $0<\alpha<1$, which is common in adaptive diffusion networks \cite{matta2018consistent}.
Considering different stabilities, it holds that
\begin{equation}
\mathop {\lim }\limits_{t \to \infty } W^t=\left \{
\begin{aligned}
&\bm{0},~&&\text{if}~W \in \mathcal{S}_a, \\
&W^{\infty},~&& \text{if}~W \in \mathcal{S}_m,
\end{aligned}\right.
\end{equation}
where $\bm{0}$ represents all-zero matrix in compatible dimensions and $\|W^{\infty}\|<\infty$.
In a recursive form, (\ref{eq:global_model}) is rewritten as
\begin{align} \label{eq:expand_form}
y_{t}=x_t+\upsilon_t=W^{t}x_{0}+\sum\limits_{m = 1}^{t} W^{m-1} \theta_{t-m} + \upsilon_t.
\end{align}
By starting the NS \eqref{eq:global_model} and letting it evolve with time, a single trajectory of the system is collected, and we organize the states/observations/noises from time $0$ to $T$ as
\begin{equation}
\begin{aligned}
X_T^- &=[x_{0},x_{2},\cdots,x_{T-1}] ,~X_T^+ =[x_{1},x_{2},\cdots,x_{T}] , \\
Y_T^- &=[y_{0},y_{2},\cdots,y_{T-1}] ,~Y_T^+ =[y_{1},y_{2},\cdots,y_{T}] , \\
\Theta_T &=[\theta_0,\theta_1\cdots,~\theta_{T-1}],~\Upsilon_T =[\upsilon_1,\upsilon_2\cdots,\upsilon_{T}] .
\end{aligned}
\end{equation}
Then, the whole evolution process is compactly written as
\begin{equation}
\begin{aligned}
X_T^+ = W X_T^- +\Theta_T,~Y_T^+ = W X_T^+ +\Upsilon_T.
\end{aligned}
\end{equation}

\subsection{Basic Inference Principles and Problem of Interest}
During the running process of the NS, the system states become highly correlated after continuous exchange of information. 
Therefore, the connectivity between two nodes can be revealed by the state correlation. 
From this perspective, the famous Pearson correlation coefficient provides a way to quantify the correlation degree, given by 
\begin{equation}\label{eq:c_coefficient}
\varrho_{ij}=\sum\limits_{t = 0}^{T} \frac{(x_t^{i} - \bar x^{i} )}{\varrho_{i}} \frac{(x_t^{j}- \bar x^{j} )}{\varrho_{j}},
\end{equation}
where $\varrho_{i}=\sqrt{\sum\nolimits_{t = 0}^{T} (x_t^{i} - \bar x^{i} )^2 }$ is the sample standard deviations of $\{x_t^i\}_{t=0}^{T}$ and $\bar x^{i}=\sum\nolimits_{t = 0}^{T} x_t^{i}/T$, $\forall i\in \mathcal{V}$. 
The larger $\varrho_{ij}$ is, the more confident one can determine that there exists an edge between node $i$ and $j$. 

Note that the coefficient $\varrho_{ij}$ directly describes the (linear) correlation between two nodes. 
However, due to its symmetry, it cannot reveal the directionality (i.e., causality) of an existing edge between two nodes. 
The following lemma presents a way to overcome the causality issue. 
\begin{lemma}[Granger causality \cite{granger1969investigating,santos2019local}]\label{le:Granger}
If multiple trajectories are available over the system  (\ref{eq:global_model}), then we have 
\begin{equation}
{R_1^x}(t)  = WR_0^x(t\!-\!1), 
\end{equation}
where $R_0={\mathbb{E}}\left[ {{x_{t}} x_{t}^{\mathsf{T}}} \right]$ and $R_1={\mathbb{E}}\left[ {{x_{t}} x_{t-1}^{\mathsf{T}}} \right]$ are the autocorrelation and one-lag autocorrelation matrices. 
\end{lemma}

This result is straightforward since ${R_1^x}(t) \!\!=\!\! {\mathbb{E}}\left[ {{x_{t}} x_{t-1}^{\mathsf{T}}} \right] \!=\! \mathbb{E}[(Wx_{t\!-\!1}\!+\!\theta_{t})x_{t-1}^{\mathsf{T}}] \!=\! WR_0^x(t\!-\!1)$. 
Note that $R_0(t)$ can be explicitly represented as 
\begin{equation}
{R_0^x}(t)  = W^{t} x_{0}x_{0}^{\mathsf{T}} (W^{t})^\mathsf{T} + \sigma_{\theta}^2 \sum\limits_{m = 0}^{t-1} W^m (W^m)^\mathsf{T}. 
\end{equation}
According to Lemma \ref{le:Granger}, the Granger estimator is given by 
\begin{itemize}
\item \textbf{Granger estimator}:
\begin{equation}\label{Granger_estimator}
\widehat{W}_g=R_1^x(t)(R_0^x(t-1) )^{-1}.
\end{equation}
\end{itemize}
Note that the construction of (\ref{Granger_estimator}) can be interpreted as using the states at the same $(t-1)$-th and $t$-th moments from sufficient trajectories of the system.

Next, we present the popular OLS estimator, which is derived from least squares optimization. 
Then, inferring the $W$ from $\{y_t\}_{t=0}^T$ is formulated to solve the following problem
\begin{equation}
\begin{aligned}
\textbf{P}_\textbf{1}:~~~~
\mathop {\min }\limits_{W}\sum\limits_{t = 1}^{T} \| y_t-W y_{t-1} \|^2.
\end{aligned}
\end{equation}
Note that the objective function of $\textbf{P}_\textbf{1}$ can be rewritten as $\mathop {\min }\limits_{W} \| Y_T^{+} - W Y_T^- \|_F^2$. 
Then, by finding the derivative, one obtains the optimal solution as 
\begin{itemize}
\item \textbf{OLS estimator}:
\begin{equation} \label{OLS_estimator}
\widehat{W}_o \!= \!Y_T^+ (Y_T^-)^\mathsf{T} (Y_T^- (Y_T^-)^\mathsf{T})^{-1}.
\end{equation}
\end{itemize}

Based on the above formulation, the prime goal of this paper is to design an efficient topology estimator for a single trajectory by revealing the relationships of the inference principles from multiple to single trajectory, and characterize the non-asymptotic inference performance in terms of convergence and accuracy. 
To practice, we establish an interpretable inference estimator borrowing the idea of node causality and correlation. 
Then, a probability analysis framework is employed for the inference performance analysis by resorting to the concentration measure in Gaussian space.

\section{The Proposed Topology Inference Method}\label{main-results}
In this section, we propose a causality-based inference method for a single trajectory setting, followed by its correlation-based modification design for cases when the observation size is small. 

\subsection{Causality-based Inference Method}
Although the Granger estimator presents a direct and analytic expression for inferring $W$, it is based on observations over multiple trajectories and the observation noises are often ignored.
It cannot be directly applied in single observation trajectory case.
Nevertheless, it provides beneficial modeling ideas from the perspective of node causality.
Similar with ${R_0^x}(t)$ and ${R_1^x}(t)$, we define the following sample covariance matrix and its one-lag version as
\begin{equation}\label{eq:sample11}
\begin{aligned}
\Sigma_0(T)=\frac{1}{T}(Y_T^-) (Y_T^-)^\mathsf{T},~\Sigma_1(T)=\frac{1}{T}(Y_T^+) (Y_T^-)^\mathsf{T}.
\end{aligned}
\end{equation}
Before demonstrating the connection between $R_0^x(R_1^x)$ and $\Sigma_0(\Sigma_1)$, we first present the following lemma.
\begin{lemma}[Mutual independence between states and noises in a single trajectory] \label{le:noise-state}
Given arbitrary $X_T\in \mathbb{R}^{n \times T}$ and noise matrix $\Theta_T \in \mathbb{R}^{n \times T}$ with i.i.d. zero-mean Gaussian entries, and let $|x|_{m}\!=\!\max\{|x_t^j|: t \!\in\!\mathbb{N}^+,j\in\mathcal{V}\}$. 
If $|x|_{m}<\infty$, then 
\begin{align}\label{limit_pro}
\Pr \left \{\mathop {\lim }\limits_{T \to \infty } \frac{1}{T} \Theta_T X_T^\mathsf{T} = \bm{0} \right\}=1.
\end{align}
\end{lemma}
\begin{proof}
The proof is provided in Appendix \ref{pr:noise-state}. 
\end{proof}

Lemma \ref{le:noise-state} illustrates the independence of the sample matrix on the noise matrix in a single observation trajectory.
The result (\ref{limit_pro}) also applies to any linear transform ${M}\Theta$ (here ${M}\in\mathbb{R}^{n\times n}$ and $\|{M}\| <\infty$).
Since only $\{y_t\}_{t=0}^{T}$ are directly available, for two consecutive observations, it follows that
\begin{align}\label{eq:two-observation}
y_t&=Wx_{t-1}+\theta_{t-1}+\upsilon_t \nonumber \\
&=Wy_{t-1}-W\upsilon_{t-1}+\theta_{t-1}+\upsilon_t \nonumber \\
&=Wy_{t-1}+\omega_t,
\end{align}
where $\omega_t=-W\upsilon_{t-1}+\theta_{t-1}+\upsilon_t$, satisfying $N(0,\sigma_{\upsilon}^2 WW^\mathsf{T} + \sigma_{\upsilon}^2 I + \sigma_{\theta}^2 I)$, which is highly auto-correlated.
Besides, $\omega_t$ is independent of all $\{x_{\tau}\}_ {\tau< t}$ and $\{\theta_{\tau}\}_ {\tau< t-1}$. 
Note that (\ref{eq:two-observation}) only represents the quantitative relationship between consecutive observations, not a causal dynamical process.
Based on the formulation, we present the following theorem.


\begin{theorem}[Causality in single observation trajectory]\label{th:causality_estimator}
Given observations $\{ y_{t}\}_{t=1}^{T}$, if $W\in \mathcal{S}_a$, we have
\begin{equation}\label{eq:conclusion2}
\Sigma_1(\infty)=W (\Sigma_0(\infty) -\sigma_{\upsilon}^2I),
\end{equation}
where $\Sigma_1(\infty)=\mathop {\lim }\limits_{T \to \infty }\Sigma_1(T)$ and $\Sigma_0(\infty)=\mathop {\lim }\limits_{T \to \infty }\Sigma_0(T)$.
\end{theorem}
\begin{proof}
The proof is provided in Appendix \ref{pr:causality_estimator}. 
\end{proof}

Different from the Granger causality in Lemma \ref{le:Granger}, Theorem \ref{th:causality_estimator} relaxes the dependence on multiple trajectories, and presents the observation causality in a single trajectory, while taking the observation noises into consideration. 
Then, given finite horizon $T$, we propose the causality-based estimator as  
\begin{itemize}
\item \textbf{Causality-based estimator}:
\begin{equation}\label{causality_estimator}
\widehat{W}_c\!=\!\Sigma_1(T) (\Sigma_0(T) -\sigma_{\upsilon}^2I)^{-1}.
\end{equation}
\end{itemize}

\begin{remark}
We demonstrate that although the estimator $\widehat{W}_c$ is derived from Theorem \ref{th:causality_estimator} where $W \in \mathcal{S}_a$ holds, it is also applicable when $W \in \mathcal{S}_m$. 
In fact, Theorem \ref{th:causality_estimator} is directly based on the Chebyshev inequality, where the bounded state constraint precludes us from proving the convergence and accuracy of $\widehat{W}_c$ when $W \in \mathcal{S}_m$. 
To tackle this issue, we can resort to the concentration measure in Gaussian space. 
The details will be given in Section \ref{section:Non-asymptotic}. 
\end{remark}


\subsection{Correlation-based Modification Design}
The proposed causality-based estimator $\widehat{W}_c$ is motivated by asymptotic relationship \eqref{eq:conclusion2}. 
When $T$ is small, directly using $\widehat{W}_c$ may be inappropriate. 
Inspired by the correlation measurement (\ref{eq:c_coefficient}), an alternative way to alleviate the influence of observation noises is to implement correlation coefficient calculation, which captures the linear correlation between nodes.
Then, we define the following correlation-based sample matrix and its one-lag version as
\begin{equation}\label{eq:sample22}
S_0(T)=\frac{1}{T}\sum\limits_{t = 0}^{T-1} \tilde y_{t}^- (\tilde y_{t}^-)^\mathsf{T},~S_1(T)=\frac{1}{T}\sum\limits_{t = 1}^{T} \tilde y_{t}^+ (\tilde y_{t-1}^-)^\mathsf{T},
\end{equation}
where the elements of $\tilde y_{t}^{-}$ and $\tilde y_{t}^{+}$ are given by
\begin{equation}\label{eq:filtered_y}
[\tilde y_{t}^{-}]^i= [y_t-\frac{ Y_T^- \bm{1} }{T}]^i/\varrho_i^-, ~[\tilde y_{t}^{+}]^i= [y_t-\frac{ Y_T^+ \bm{1} }{T}]^i/\varrho_i^+.
\end{equation}
Here the correlation coefficients are computed by 
\begin{align}
\varrho_i^- =\sqrt{\sum\limits_{t = 0}^{T-1} (y_t^{i} - [\frac{ Y_T^- \bm{1} }{T}]^i )^2 }, \varrho_i^+ =\sqrt{\sum\limits_{t = 1}^{T} (y_t^{i} - [\frac{ Y_T^+ \bm{1} }{T}]^i )^2 }. \nonumber
\end{align}
Finally, the correlation-based modified version of the proposed $\widehat{W}_c$ is designed as
\begin{itemize}
\item \textbf{Correlation-modified estimator}:
\begin{equation}\label{correlation_estimator}
\widehat{W}_s = S_1(T)S_0^{-1}(T) .
\end{equation}
\end{itemize}


The main merit of $\widehat{W}_s$ lies in it takes the noise filtering and the node correlation into account at the same time.
In statistics, it can be seen as a normalization operation to quantity the observations in the same measurement space. 
We point out that correlation-based modification $\widehat{W}_s$ improves the inference performance of estimator $\widehat{W}_c$ in small observation scale, and its inference accuracy is no worse than that of $\widehat{W}_o$ (this will be verified in section \ref{simulation}).

\begin{remark}\label{rema:invertibility}
Note that under Assumption \ref{assu:noise}, the invertibility of $\Sigma_0(T)$, along with its transformations $(\Sigma_0(T) -\sigma_{\upsilon}^2I)$ and $S_0(T)$, is guaranteed. 
On the one hand, the i.i.d. process noise $\theta_t$ ensures that the system states at different times are linearly independent of each other. 
On the other hand, the addictive i.i.d. observation noise $\upsilon_t$ further enhances the linear independence of the accessed observations. 
According to Sard's theorem in measure theory, the matrices $\Sigma_0(T)$, $(\Sigma_0(T) -\sigma_{\upsilon}^2I)$ and $S_0(T)$ are  full-ranked almost surely. 
\end{remark}

\begin{figure}[t]
\centering
\includegraphics[width=0.42\textwidth]{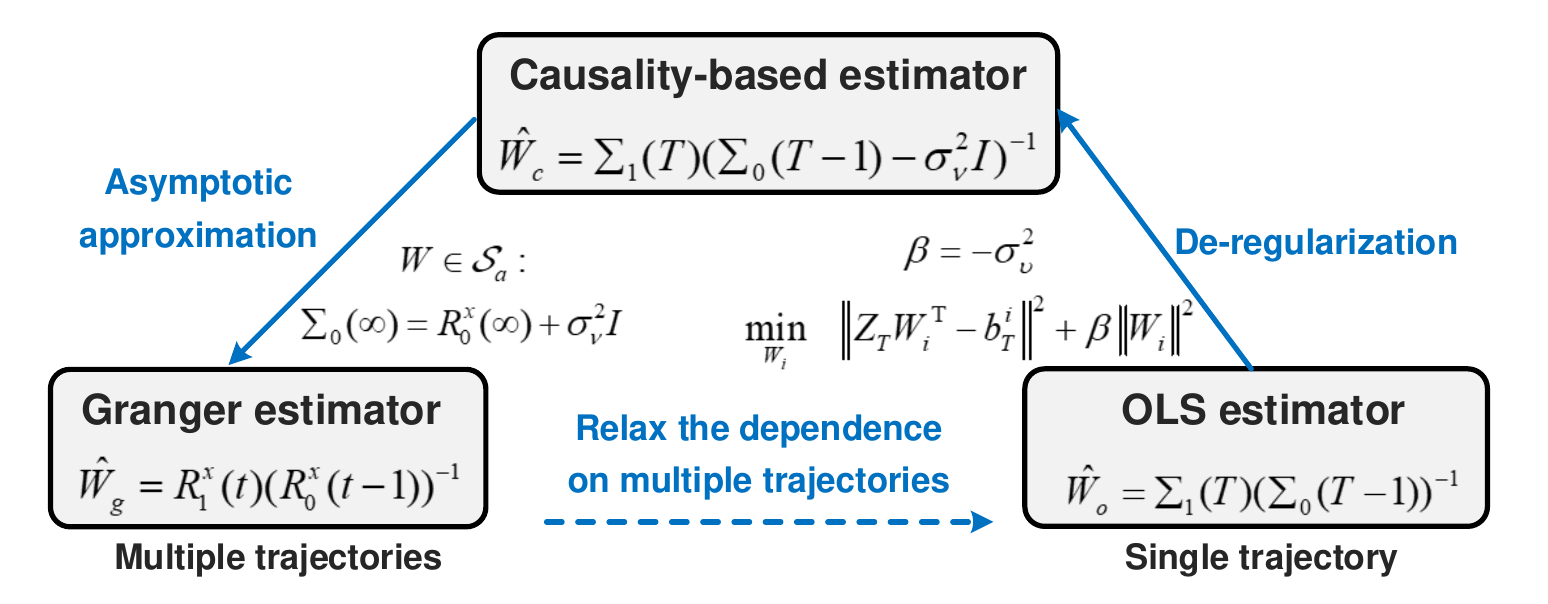}
\vspace{-5pt}
\caption{An overview of the mutual relationships of $\widehat{W}_c$, $\widehat{W}_g$ and $\widehat{W}_o$. 
}
\label{relationships}
\vspace{-10pt}
\end{figure}

\subsection{Relationships between Different Estimators}
In this part, we demonstrate the relation between the causality-based estimator $\widehat{W}_c$, the Granger estimator $\widehat{W}_g$, and OLS estimator $\widehat{W}_o$. 
\begin{theorem}[Equivalence condition between $\Sigma_0$ and $R_0$] \label{th:equivalent1}
If $W\in \mathcal{S}_a$, when $T\to\infty$, we have
\begin{equation}\label{eq:equivalent}
\Sigma_0(\infty)=R_0^x(\infty)+\sigma_{\upsilon}^2 I,~\Sigma_1(\infty)=R_1^x(\infty)+\sigma_{\upsilon}^2 W.
\end{equation}
\end{theorem}
\begin{proof}
The proof is provided in Appendix \ref{pr:equivalent1}. 
\end{proof}

Theorem \ref{th:equivalent1} demonstrates the equivalent condition between estimators $\widehat{W}_g$ and $\widehat{W}_c$.
It reveals that the expected state covariance matrix of $T\to \infty$ is identical to the sample covariance matrix along all the single time horizon, which is an interesting result that describes the relationship between multiple and single observation trajectories.

Note that both $\widehat{W}_c$ and $\widehat{W}_o$ can be computed row-by-row in a parallel manner. 
To fit the common least squares problem modeling, we interpret the relationships of the two estimators by focusing on an arbitrary row of them. 
First, define the following quadratic fractional optimization problem
\begin{equation}\label{eq:Rayleigh}
\begin{aligned}
\textbf{P}_\textbf{2}:~\min_{W_{i} \in \mathbb{R}^{1\times n}} \frac{\left\| Z_T W_{i}^{\mathsf{T}} -\bm{b}_{T}^{i}\right\|^2 }{\left\|W_{i}^{\mathsf{T}} \right\|^2 +1},
\end{aligned}
\end{equation}
where $Z_T=\frac{(Y_{T}^-)^{\mathsf{T}}}{\sqrt{T}} \in \mathbb{R}^{T \times n}$ is the coefficient matrix, and  
$\bm{b}_{T}^{i}=[y_{1}^{i},y_{2}^{i},\cdots,y_{T}^{i}]^\mathsf{T}/{\sqrt{T}} \in \mathbb{R}^{T}$ is the scaled observations of node $i$ from $0$ to $T-1$. 
The objective function in \eqref{eq:Rayleigh} also known as Rayleigh quotient, and the numerator in \eqref{eq:Rayleigh} is exactly the least squares for obtaining $\widehat{W}_{i,o}$. 
Therefore, $\textbf{P}_\textbf{2}$ can be regarded as a weighted version of OLS problem. 
Let $\widehat{W}_{i,c}$ represent the $i$-row of $\widehat{W}_{c}$, and $\widehat{W}_{i,\tilde{o}}$ be the solution of $\textbf{P}_\textbf{2}$. 
Then, we have the following theorem. 

\begin{theorem}[Relationship between $\widehat{W}_c$ and $\widehat{W}_{i,\tilde{o}}$]\label{th:relationship}
In a single trajectory, when $T\to\infty$, $\widehat{W}_{i,c}$ is equivalent to $\widehat{W}_{i,\tilde{o}}$, i.e., 
\begin{equation}\label{eq:OLS_equivalent}
\mathop {\lim }\limits_{T \to \infty } \widehat{W}_{i,c}(T) = \mathop {\lim }\limits_{T \to \infty } \widehat{W}_{i,\tilde{o}}(T). 
\end{equation}
\end{theorem}

\begin{proof}
The proof is provided in Appendix \ref{pr:relationship}.
\end{proof}

Based on Theorem \ref{th:relationship}, the relationships between $\widehat{W}_c$ and $\widehat{W}_o$ can be interpreted from two aspects. 
First, $\widehat{W}_{i,c}$ is the best asymptotic estimator that minimize weighted least squares by multiplying $\frac{1}{\left\|W_{i}^{\mathsf{T}} \right\|^2 +1}$. 
It provides a new interpretation for revising OLS methods to infer the topology using noisy observation from causality perspective. 
Second, $\widehat{W}_{i,c}$ is a de-regularization version of $\widehat{W}_{i,o}$ by obviating the influence of observation noises. 
Specifically, one can easily obtain $\widehat{W}_{i,c}$ by finding the stationary point of the following problem 
\begin{equation}\label{eq:de-regulaization}
\min_{W_{i} \in \mathbb{R}^{1\times n}}~g(W_{i,c})=\left\| Z_T W_{i}^{\mathsf{T}} -\bm{b}_{T}^{i}\right\|^2-\sigma_{\upsilon}^2\left\| W_{i}^{\mathsf{T}} \right\|^2,
\end{equation}
where $-\sigma_{\upsilon}^2\left\| W_{i}^{\mathsf{T}} \right\|^2$ can be regarded as a de-regularization term\footnote{In the literature, considering the optimization problem $\mathop {\min }\limits_{x} \| Z_{T}x- b_{T}\|^2$ can be ill-posed, one can add some penalty term about $x$ in the objective function (e.g., Tikhonov term $\beta \| x \|_2$ or Lasso term $\beta \| x \|_1$ with $\beta>0$), and this technique is called regularization and $\beta$ is called regularization coefficient. 
Here in \eqref{eq:de-regulaization}, the added term is associated with a negative coefficient, and thus we call it de-regularization.}. 
Notice that here we mention $\widehat{W}_{i,c}$ is an stationary point because $g(W_{i})$ is not necessarily convex.

\begin{remark}
As indicated in \eqref{eq:de-regulaization}, the original topology matrix can be estimated row-by-row. 
Based on this independent manner, the i.i.d. Gaussian noises for simple analysis can be easily relaxed to independent but non-identical cases, i.e., $\mathbb{E} \upsilon_{t_1} \upsilon_{t_2}^{\mathsf{T}}=\delta_{t_1 t_2}\operatorname{diag}(\sigma_{\upsilon_1}^2, \sigma_{\upsilon_2}^2,\cdots,\sigma_{\upsilon_n}^2)$. 
Consequently, this relaxation will not affect the non-asymptotic inference performance as long as $\max\{\sigma_{\upsilon_1}^2, \sigma_{\upsilon_2}^2,\cdots,\sigma_{\upsilon_n}^2\}$ is strictly bounded. 
Furthermore, even for the cases where $\sigma_{\upsilon}^2$ is not prior known, one can turn to solve its alternative problem $\textbf{P}_\textbf{2}$ based on Theorem \ref{th:relationship}. 
In this sense, the noise variance can be regarded as being indirectly estimated and the solution will approximate the proposed estimator asymptotically. 
\end{remark}

In summary, the four estimators $\widehat{W}_g$, $\widehat{W}_o$, $\widehat{W}_c$ and $\widehat{W}_s$ approximate $W$ from different angles, as depicted in Fig.~\ref{relationships}. 
From a statistical viewpoint, $\widehat{W}_g$ implements the inference using observations at identical moments in multiple trajectories, while the remaining three use a sequence of observations in a single trajectory, which is more common in practice.
Specifically, $\widehat{W}_s$ is a modified version of $\widehat{W}_c$ for small horizon $T$, whose inference accuracy is no worse than that of $\widehat{W}_o$.

\section{Inference Performance Analysis: \\Convergence and Accuracy}\label{section:Non-asymptotic}

In this section, we analyze the non-asymptotic inference performance of the proposed causality-based estimator $\widehat{W}_c$ in terms of convergence speed and accuracy. 


To begin with, we provide the supplementary results of Theorem \ref{th:equivalent1} when $W\in \mathcal{S}_m$, by clarifying the non-asymptotic deviation of the observation matrices used in $\widehat{W}_c$ and $\widehat{W}_g$. 
\begin{lemma}[Concentration measure in Gaussian space \cite{davidson2001local}]\label{le:single-value}
Let $\Theta\in \mathbbm{R}^{n \times T}$ be a matrix with independent standard normal entries. 
With probability at least $1-2 \exp \left(-r^{2} / 2\right)$, the singular values of $\Theta$ satisfy 
\begin{equation}
\sqrt{T}-\sqrt{n}-r \leq \rho_{\min}(\Theta) \leq \rho_{\max}(\Theta) \leq \sqrt{T}+\sqrt{n}+r. 
\end{equation}
\end{lemma}


\begin{theorem}[Sample matrix deviation between $\Sigma_0$ and $R_0$]\label{th:equivalent-nonstable}
If $W\in \mathcal{S}_m$, the deviation norm $\| \Sigma_0(T) -R_0^x(T) - \sigma_{\upsilon}^2 I \|$ is at least in $\bm{O}(\sqrt{T})$ scale. 
\end{theorem}
\begin{proof}
The proof is provided in Appendix \ref{pr:equivalent-nonstable}. 
\end{proof}

Theorem \ref{th:equivalent-nonstable} reveals that when $W\in \mathcal{S}_m$, the influence of the process noises will consistently accumulate as $T$ increases, and one cannot use the $\Sigma_0(\infty)$ to approximate the ideal factor $R_0^x(\infty)$. 
However, we will demonstrate this defect does not hinder us from using the estimator. 
The key question here is what is the exact influence of whether $W\in \mathcal{S}_m$ or $W\in \mathcal{S}_a$ over the inference performance. 
A direct intuition is that it needs extra cost to overcome the accumulated influence of process noises when $W\in \mathcal{S}_m$. 
To analyze this issue, we define 
\begin{equation}
\Sigma_T= T \Sigma_0(T-1)=Y_{T}^{-} (Y_{T}^{-})^\mathsf{T} ,
\end{equation}
and then introduce the following lemma. 
\begin{lemma}[Proposition 3.1 in \cite{sarkar2019near}]\label{le:lsm-bound1}
Let $V \succ 0$ be a deterministic matrix and ${\tilde\Sigma}_{T}=\Sigma_{T} + V $. 
Given $0<\delta<1$ and $\{\upsilon_t,y_t\}_{t=1}^{T}$ defined as before, we have with probability $1-\delta$
\begin{align}
\left\| (\sum_{t=0}^{T-1} \upsilon_{t+1} y_{t}^\mathsf{T} ) \tilde\Sigma_{T}^{-\frac{1}{2}} \right\|
\!\leq\! \sqrt{8 n \log \left(\frac{5 \operatorname{det}\left({\tilde\Sigma}_{T} V^{-1}\!+\!I\right)^{  \frac{1}{2n} }}{\delta^{\frac{1}{n}}}\right) }. 
\end{align}
\end{lemma}

Lemma \ref{le:lsm-bound1} shows the existence of the upper bound for $\left\| (\sum_{t=0}^{T-1} \upsilon_{t+1} y_{t}^\mathsf{T} ) \tilde\Sigma_{T}^{-\frac{1}{2}} \right\|$, and the invertibility of $\tilde\Sigma_{T}$ is where most of the proof lies. 
In Remark \ref{rema:invertibility}, the invertibility of $\Sigma_{T}$ demonstrated. 
Therefore, when $T$ is sufficient large, we can always find deterministic $V_{dn}$ and $V_{up}$ such that 
\begin{equation}\label{eq:up-down}
0 \prec V_{dn} \preceq  \Sigma_T  \preceq \Sigma_{T} +V_{dn} \preceq V_{up}
\end{equation}
holds with high probability. 
Following Lemma \ref{le:lsm-bound1} and (\ref{eq:up-down}), we present the non-asymptotic bound of $\widehat{W}_o$, 
paving the way for subsequent comparisons. 
\begin{theorem}[Error bound by $\widehat{W}_o$]\label{th:asymptotic-performance}
Given $\{y_t,\theta_t,\upsilon_t\}_{t=0}^{T}$ defined before, with probability at least $1-\delta$, 
the following non-asymptotic bound holds,
\begin{small}
\begin{equation}\label{eq:estimate_bound}
\| \widehat{W}_o-W \| \leq \frac{ 12 \sqrt{n \log\left(\frac{5 \operatorname{det} (  V_{up} V_{dn}^{-1}+ I)^{ \frac{1}{2n}}  } {\delta^{1 / n}} \right)}   + \frac{5T\sigma_{\upsilon}}{4 \sqrt{ \lambda_{\min}(V_{dn}) }} }{\sqrt{\lambda_{\min}(V_{dn})}} .
\end{equation}
\end{small}
\end{theorem}
\begin{proof}
The proof is provided in Appendix \ref{pr:asymptotic-performance}. 
\end{proof}

Theorem \ref{th:asymptotic-performance} demonstrates that the non-asymptotic performance is mainly determined by $\lambda_{\min}(V_{dn})$ and $\|V_{up}\|$. 
The non-asymptotic bound of $\widehat{W}_c$ is in the same form as that of $\widehat{W}_o$ and is omitted here. 
It is straightforward that if the term $\sqrt{\lambda_{\min}(V_{dn})}$ grows faster than the numerator in (\ref{eq:estimate_bound}) as $T$ increases, the inference accuracy also increases. 
Next, we explicitly characterize the convergence and accuracy of the two estimators. 
\begin{theorem}[Convergence speed and accuracy of $\widehat{W}_o$ and $\widehat{W}_c$]\label{th:converge-speed}
With probability at least $1-\delta$, the non-asymptotic bound of the OLS estimator $\widehat{W}_o$ satisfies 
\begin{equation} \label{eq:wo_bound}
\| \widehat{W}_o-W \| \sim\left \{
\begin{aligned}
&\bm{O}(\sqrt{ \frac{ \log{T} }{T} })+\bm{O}(\sigma_{\upsilon}^2),~&& \text{if}~W \in \mathcal{S}_m, \\
&\bm{O}(\frac{1}{\sqrt{T}})+\bm{O}(\sigma_{\upsilon}^2),~&&\text{if}~W \in \mathcal{S}_a. 
\end{aligned}\right.
\end{equation}
and the non-asymptotic bound of the proposed $ \widehat{W}_c $ satisfies 
\begin{equation}
\| \widehat{W}_c-W \| \sim\left \{
\begin{aligned}
&\bm{O}(\sqrt{ \frac{ \log{T} }{T} }),~&& \text{if}~W \in \mathcal{S}_m, \\
&\bm{O}(\frac{1}{\sqrt{T}}),~&&\text{if}~W \in \mathcal{S}_a. 
\end{aligned}\right.
\end{equation}
\end{theorem}
\begin{proof}
The proof is provided in Appendix \ref{pr:converge-speed}. 
\end{proof}

Theorem \ref{th:converge-speed} demonstrates the convergence rates of the inference error bounds by using $\widehat{W}_o$ and $\widehat{W}_c$. 
Now back to the question before Lemma \ref{le:lsm-bound1}, we can conclude that the extra cost for the estimators when $W\in \mathcal{S}_m$ is longer converging time (or larger observation number), requiring $\bm{O}(\sqrt{  \log{T} })$ times than that when $W\in \mathcal{S}_a$. 
In terms of accuracy, when $T\to \infty$, the inference error will converge to a constant by $\widehat{W}_o$, while that of $\widehat{W}_c$ will converge to zero.

\section{Extended Discussions}\label{section:extension}

\subsection{Online/Recursive Version of the Causality-based Estimator}

Note that although $g(W_{i,c})$ describes the average error for estimating $W_{i}$, the stationary point $\widehat{W}_{i,c}$ will not change if we multiply the observation number $T$ with $g(W_{i,c})$. 
Following this implication, we define 
\begin{align}
\tilde{Z}_T= \sqrt{T}Z_T=[\tilde{z}_1^\mathsf{T},\tilde{z}_2^\mathsf{T},\cdots,\tilde{z}_{T}^\mathsf{T}]^\mathsf{T} \in \mathbb{R}^{T \times n},\\
\bm{\tilde{b}}_{T}^{i}= \sqrt{T}\bm{b}_{T}^{i}=[\tilde{b}_{1}^{i},\tilde{b}_{2}^{i},\cdots,\tilde{b}_{T}^{i}]^\mathsf{T} \in \mathbb{R}^{T \times 1},
\end{align}
and present an online/recursive version of $\widehat{W}_{i,c}$ as follows. 

\begin{corollary}\label{coro:recursive}
Given historic estimates $\widehat{W}_{i,c}(t-1)$ and $P_{t-1}$ till time $t-1$ $( t\le T)$, when the latest observations $\tilde{z}_{t}$ and $\tilde{b}_{i,t}$ are supplied, $\widehat{W}_{i,c}(t)$ can be recursively computed by
\begin{small}
\begin{align}
\!\!&\widehat{W}_{i,c}^{\mathsf{T}}(t)\!=\!(I \!+\! \sigma_{\upsilon}^2 P_t)\widehat{W}_{i,c}^{\mathsf{T}}(t\!-\!1)\!+\!P_{t} \tilde{z}_{t}^{\mathsf{T}}(\tilde{b}_{t}^{i}\!-\!\tilde{z}_{t} \widehat{W}_{i,c}^{\mathsf{T}}(t\!-\!1)), \\
\!\!&P_{t}\!=\!P_{t-1}-P_{t-1} {U_t} \left(\Lambda_t^{-1}+ {U_t^\mathsf{T}} P_{t-1} {U_t} \right)^{-1} {U_t^\mathsf{T}} P_{t-1},
\end{align}
\end{small}
\!\!where $U_t$ and $\Lambda_t$ are the eigenvectors and diagonalizable eigenvalues matrix of the eigenvalue decomposition $(\tilde{z}_{t}^{\mathsf{T}} \tilde{z}_{t} -\sigma_{\upsilon}^2 I )= U_t \Lambda_t U_t^\mathsf{T}$, respectively. 
\end{corollary}

\begin{proof}
The proof is provided in Appendix \ref{pr:recursive}
\end{proof}

Corollary \ref{coro:recursive} shows that the estimator $\widehat{W}_{i,c}^{\mathsf{T}}(t)$ can be computed by the weighted combination of the historic estimator $\widehat{W}_{i,c}^{\mathsf{T}}(t)$ and the prediction error $(\tilde{b}_{i,t}\!-\!\tilde{z}_{t} \widehat{W}_{i,c}^{\mathsf{T}}(t\!-\!1))$. 
Therefore, the recursive estimator does not store the whole sample matrix $\tilde{Z}_T$ but only $\widehat{W}_{i,c}(t)$ and $P_t$, which is similar to the recursive OLS estimator \footnote{ The recursive version of $\widehat{W}_{i,o}^{\mathsf{T}}(t)$ can be seen as a special case of $\widehat{W}_{i,c}^{\mathsf{T}}(t)$ when $\sigma_{\upsilon}^2$=0. 
In this situation, the eigenvalue decomposition for $\tilde{z}_{t}^{\mathsf{T}} \tilde{z}_{t} $ is not needed. 
Consequently, one can directly compute $\widehat{W}_{i,o}^{\mathsf{T}}(t)\!=\!\widehat{W}_{i,c}^{\mathsf{T}}(t\!-\!1)\!+\!P_{t} \tilde{z}_{t}^{\mathsf{T}}(\tilde{b}_{t}^{i}\!-\!\tilde{z}_{t} \widehat{W}_{i,o}^{\mathsf{T}}(t\!-\!1))$, $P_{t}\!=\!P_{t-1}-P_{t-1} \tilde{z}_{t}^{\mathsf{T}}\tilde{z}_{t} P_{t-1}/(1+\tilde{z}_{t}P_{t-1}\tilde{z}_{t}^{\mathsf{T}})$, which involves no matrix inversion.}. 
However, the computation of $\widehat{W}_{i,c}^{\mathsf{T}}(t)$ needs to update $P_t$ with an extra eigenvalue decomposition and matrix inversion operation, which is not required in recursive $\widehat{W}_{i,o}^{\mathsf{T}}(t)$. 
Note that this major difference is essentially caused by the time-varying characteristic of the (de-)regularization term, making the term $\tilde{z}_{t}^{\mathsf{T}} \tilde{z}_{t} -\sigma_{\upsilon}^2 I $ cannot be simply represented by the correlation of one vector. 
To practice, the recursion of $\widehat{W}_{i,c}^{\mathsf{T}}(t)$ can be initialized by using the batch solution of the problem when $t$ is very small. 
More simple initializations are also possible by setting $\widehat{W}_{i,c}^{\mathsf{T}}(t-1)=\bm{0}$ and $P_{0}=K_0 I$ (where $K_0$ is a large positive constant, e.g., $K_0=100$).  

Apart from the storage and computation benefits, the online estimator can also be used for cases where the topology dynamically changes with time. 
Then, one can spot the evolution trend and detection topology switches, where the topology is usually assumed to be piece-wise constant \cite{baingana2016tracking}. 
The key idea is to compute the deviation between two consecutive estimators and compare it with a preset threshold. 

\begin{remark}
In the literature, efficient alternative methods that avoid matrix inversion can be found, e.g., see \cite{rhode2014recursiveIFAC,rhode2014recursivea} and the references therein. 
The key idea of these approaches is to approximate $\widehat{W}_{i,c}(t)$ by minimizing the Rayleigh quotient, which involves singular value decomposition techniques and noise correlation estimation. 
The direction is out of the scope of this paper and the details are omitted here. 
\end{remark}

\begin{algorithm}[t]
    \caption{Infer the topology structure of nonlinear cases}
    \label{nonlinear-algo}
    \begin{algorithmic}[1]
    \REQUIRE{Observations $\{y_t\}_{t=0}^{T}$, and node set $\mathcal{V}$.}
    \ENSURE{Binary adjacency matrix estimator $\hat A=[\hat a_{ij}]_{i=1:n}^{j=1:n}$.}
    \STATE Calculate the total regression time $T_r=T-n+1$.
    \FOR {$l \gets 1$\ \textbf{to} $T_r$}
    {
        \FOR {$i\in\mathcal{V}$}
        {
          \STATE Calculate the correlation-based modified observations $[\tilde{y}_{t}^{-}]^i$ and $[\tilde{y}_{t}^{+}]^i$ by \eqref{eq:filtered_y}.
        }
        \ENDFOR
      \STATE ${\tilde {W}(l)} =\mathop {\arg \min }\limits_{ W(l)}\frac{1}{n} \sum\limits_{t = l}^{l+n-1} {\|\tilde y_{t+1}- { W(l)} \tilde y_{t} \|_2^2}$.
      \STATE Adopt $k$-means method to $\tilde {W}(l)$ and obtain its corresponding binary adjacency matrix $\tilde A(l)=[\tilde a_{ij}(l)]_{i=1:n}^{j
      =1:n}$. 
    }
    \ENDFOR
    \FOR {$i,j\in\mathcal{V}$}
    {
        \STATE $\mathcal{A}_{0}^{ij}=\{\tilde a_{ij}(l)=0:l=1,\cdots, T_r\}$, \\$\mathcal{A}_{1}^{ij}=\{\tilde a_{ij}(l)>0:l=1,\cdots, T_r\}$.
        \STATE $\hat a_{ij}=1$ if $|\mathcal{A}_{1}^{ij}| > |\mathcal{A}_{0}^{ij}| $ or $\hat a_{ij}=0$ otherwise.
    }
    \ENDFOR
    \end{algorithmic}
\end{algorithm}

\begin{figure*}[ht]
\centering
\subfigure[The sample matrix deviation]{\label{fig:sample_deviation}
\includegraphics[width=0.344\textwidth]{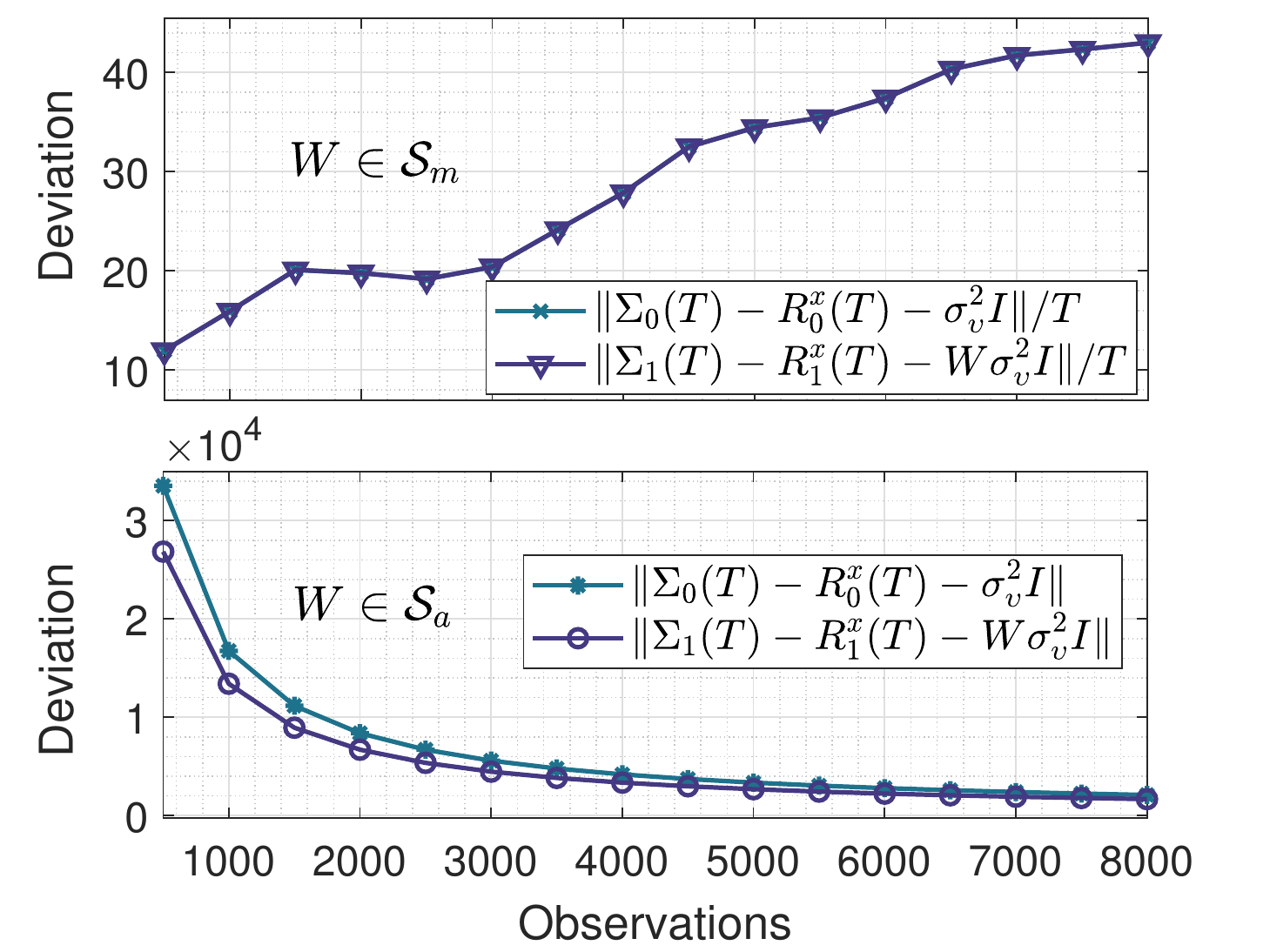}}
\hspace{-0.7cm}
\subfigure[Inference error using different estimators]{\label{fig:self_error}
\includegraphics[width=0.344\textwidth]{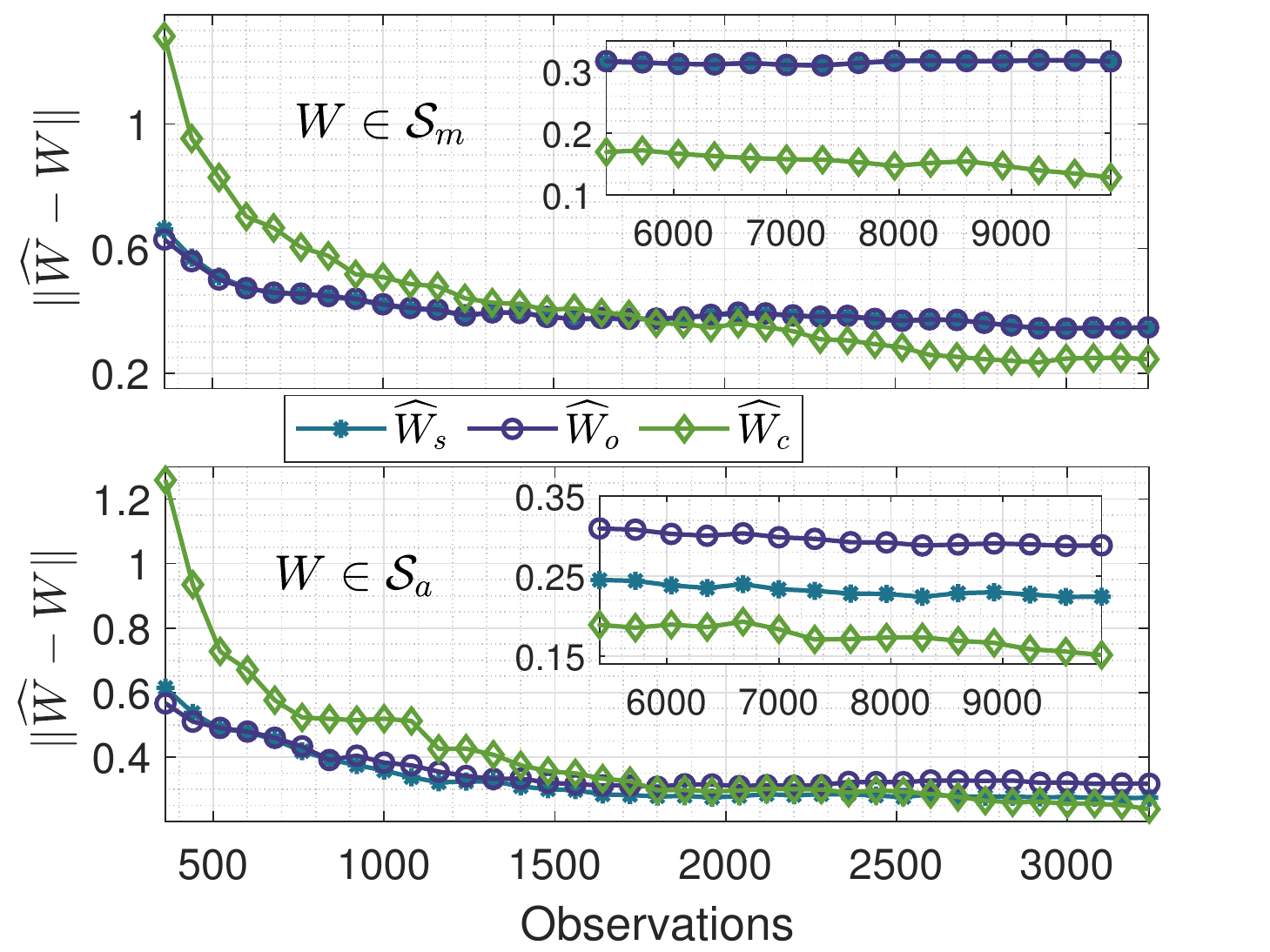}}
\hspace{-0.7cm}
\subfigure[Convergence with observation increasing]{\label{fig:convergence}
\includegraphics[width=0.344\textwidth]{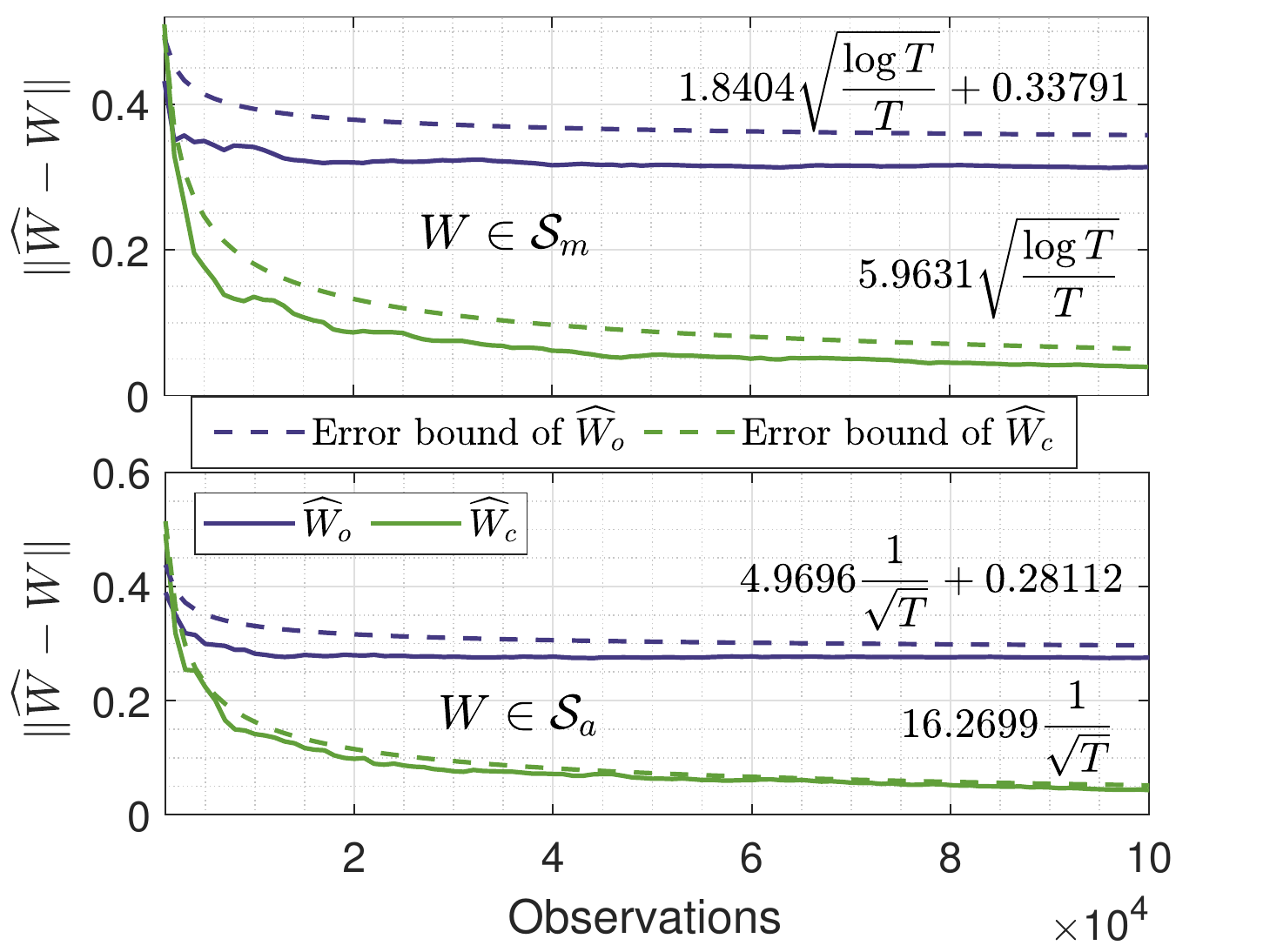}}
\caption{Experiments of verifying the theoretical results covering both asymptotically and marginally stable cases. 
(a): The sample matrix deviation in $\widehat{W}_{g}$ using multiple trajectories and $\widehat{W}_{c}$ using single trajectory. 
(b): Inference errors of $\widehat{W}_{o}$, $\widehat{W}_{c}$ and $\widehat{W}_{s}$ using the same single observation trajectory. 
(c): The convergence rates of $\widehat{W}_{o}$ and $\widehat{W}_{c}$ as observations increase. 
}
\vspace*{-10pt}
\label{fig:theoretical}
\end{figure*}

\subsection{Nonlinear Cases}
The nonlinearities of the NS model mainly come from two aspects. 
First, the magnitude of the system state cannot be unbounded,  and thus the input torque is bounded \cite{1239709}. 
Second, the edge weight in the topology matrix is not necessarily static, and it can be highly dependent on the state difference of its associated two nodes \cite{moreau2005stability}. 
Mathematically, the two kinds of nonlinearities can be uniformly formulated by 
\begin{equation}\label{eq:time-varying-rule }
x_{t+1}^{i}=x_{t}^{i}+ \sum\nolimits_{j = 1}^{n} \varphi_{ij}( x_{t}^{j}-x_{t}^{i} ), 
\end{equation}
where $\varphi_{ij}(z)$ is a continuous and strictly-bounded function, and $\varphi_{ij}(z)=0$ if $a_{ij}=0$ or $z=0$. 
As for the conditions of $\varphi_{ij}$ to guarantee the convergence and stability of the NS, the readers are referred to \cite{moreau2005stability}. 
Note that it is difficult to obtain the actual input form of each agent and find the internal edge weight. 
However, their internal adjacency structure is unchanged, which is also critical knowledge about the NS.  

Next, we will illustrate how to use our proposed revised casualty-based estimator to infer the adjacency structure. 
The key idea is as follows. 
First, we adopt linearization over a local time horizon sequentially and calculate the topology matrix by estimator (\ref{correlation_estimator}). 
Since $W$ contains $n^2$ element, at least $n+1$ groups of consecutive observations are needed to obtain a least squares solution of $W$ (suppose $y_{t+1}=W y_t$). 
Therefore, we set the local time horizon as $n+1$. 
Then, all the estimated topology matrices are integrated to discriminate whether an edge between two nodes exists by statistics. 
Specifically, a clustering procedure (e.g., $k$-means cluster method) is adopted to automatically classify the regressed weights into connected and disconnected ones. 
To strengthen the classification accuracy, a voting rule is proposed to determine the connectivity of two nodes. 
All procedures are summarized as Algorithm \ref{nonlinear-algo}. 


\begin{figure*}[t]
\centering
\subfigure[NMSE vs. observation number]{\label{test_a}
\includegraphics[width=0.344\textwidth]{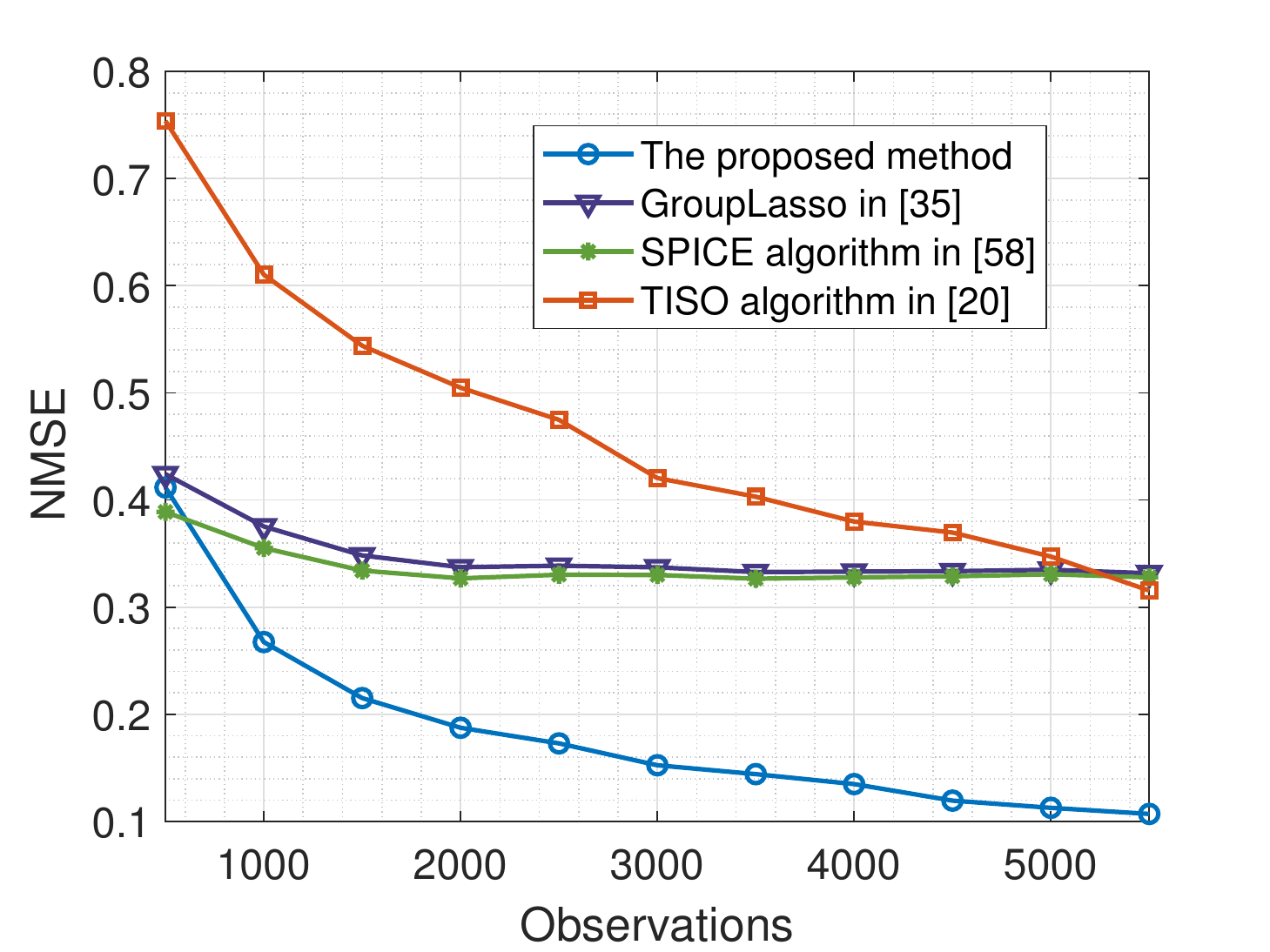}}
\hspace{-0.7cm}
\subfigure[EIER vs. observation number]{\label{test_b}
\includegraphics[width=0.344\textwidth]{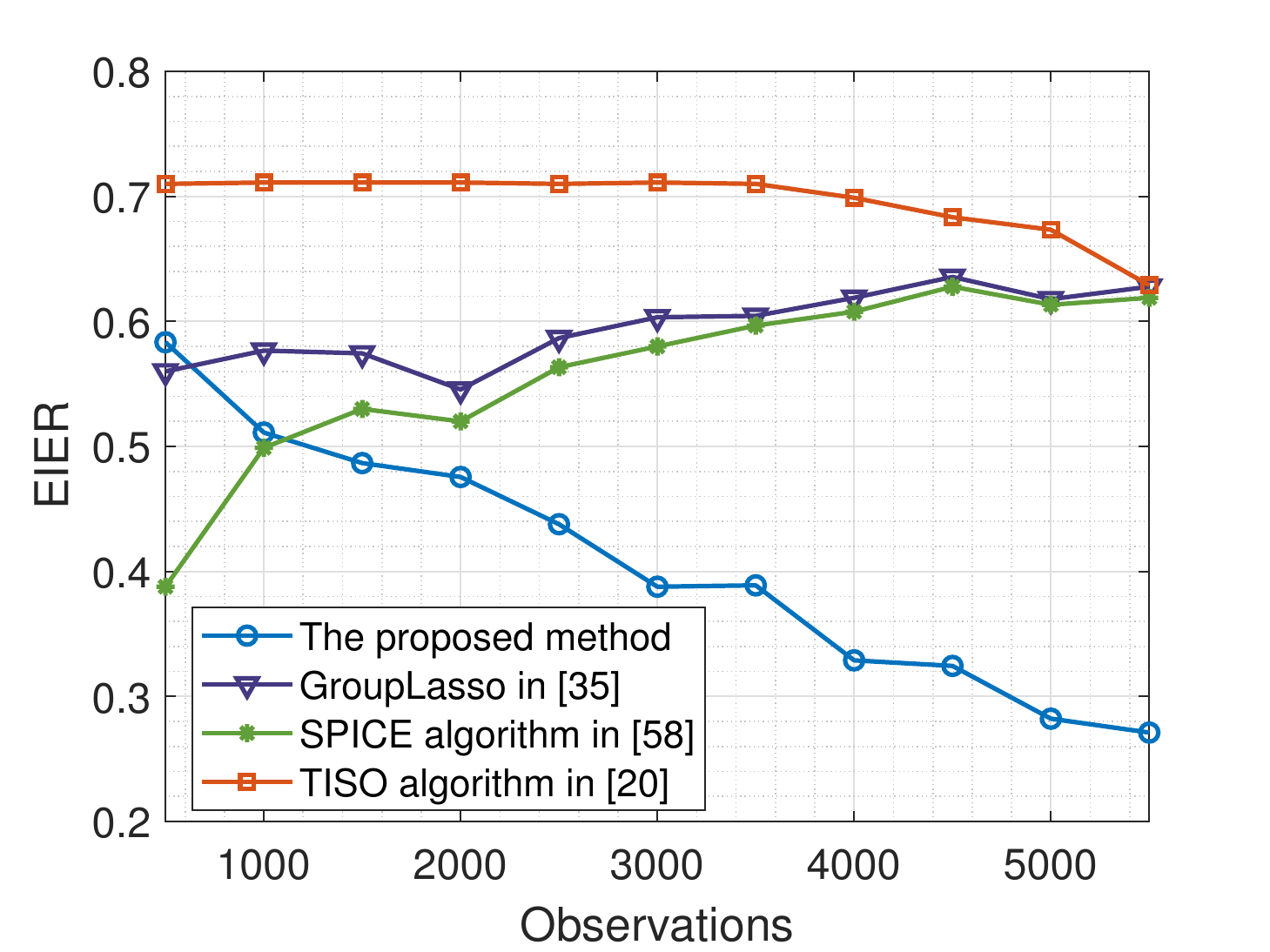}}
\hspace{-0.7cm}
\subfigure[F-score vs. observation number]{\label{test_c}
\includegraphics[width=0.344\textwidth]{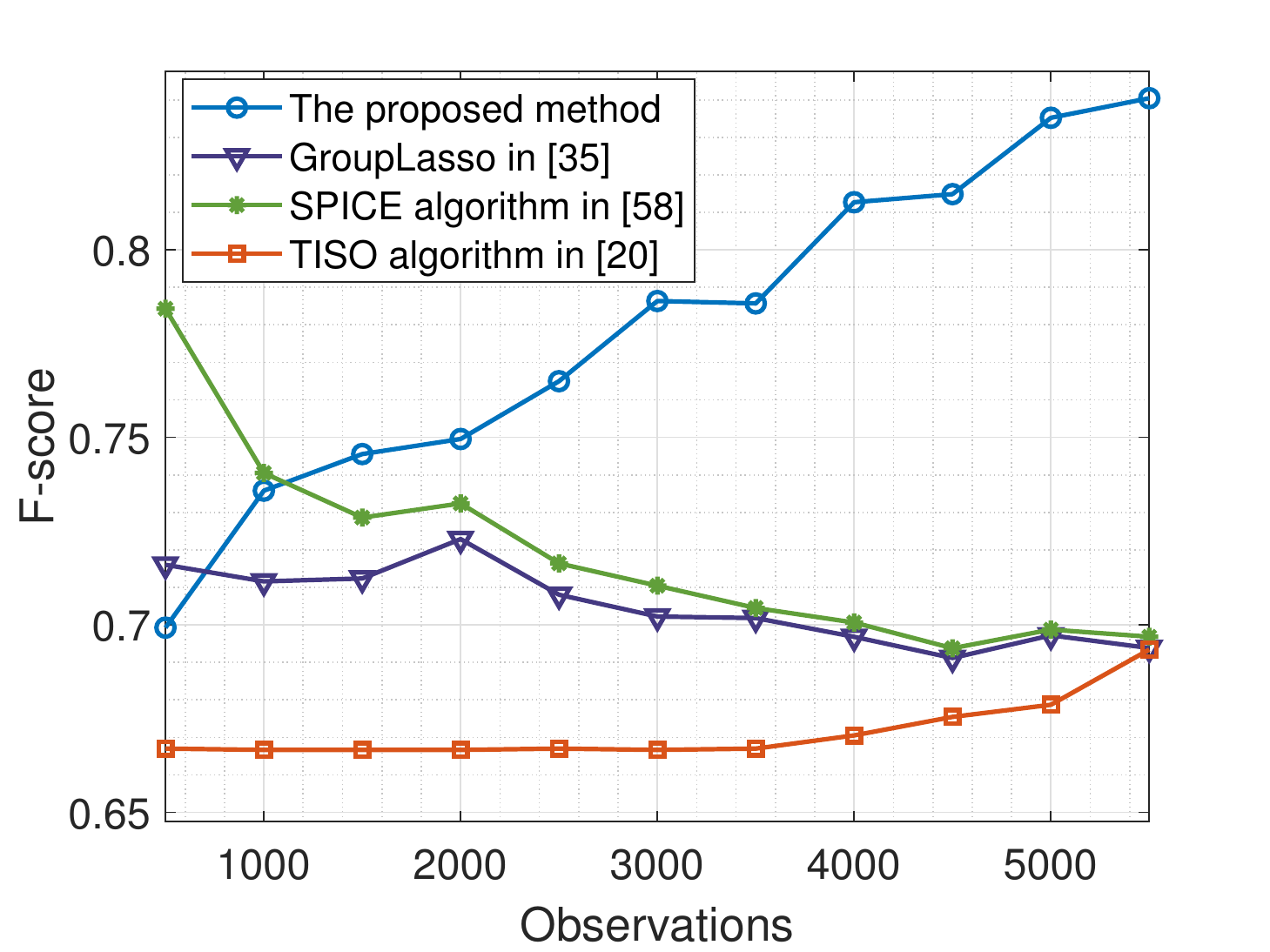}}
\caption{Comparisons of topology inference performance using the causality-based estimator, GroupLasso, SPICE, and TISO algorithms in the literature. The consecutive observations (samples) are from a single trajectory of NSs, where the variances of noises are  $\sigma_{\theta}^2=0.2$ and $\sigma_{\upsilon}^2=0.1$. }
\label{fig:com_test}
\vspace*{-10pt}
\end{figure*}

\section{Numerical Experiments}\label{simulation}
In this section, we first present numerical experiments to verify the theoretical results about the estimators' relationships and the non-asymptotic performance. 
Then, we compare the proposed causality-based estimator with some state-of-the-art methods in multiple aspects, showing its effectiveness. 
Finally, examples of nonlinear cases are provided.

\subsection{Verification of Theoretical Results}
In this experiment, we randomly generate a directed topology $W$ with $|\mathcal{V}|=20$, and the weight is designed by the Laplacian rule \eqref{eq:topo-rule}. 
Both $W\in \mathcal{S}_a$ and $W\in \mathcal{S}_m$ are considered.
For generality, the initial states of all agents are randomly selected from the interval $[400,600]$, and the variance of the process and observation noise satisfy $\sigma_{\theta}^2=1$ and $\sigma_{\upsilon}^2=1$. 

Let us begin with examining the deviation of the sample matrices used in Granger estimator $\widehat{W}_{g}$ (for multiple trajectories) and the proposed causality-based estimator $\widehat{W}_{c}$ (for single trajectory), respectively, i.e., verifying conclusions in Theorem \ref{th:equivalent1} and \ref{th:equivalent-nonstable}. 
This result is reported in Fig.~\ref{fig:sample_deviation}. 
When $W\in\mathcal{S}_a$, the sample matrix $R_0^x(T)$ from multiple observation trajectories can be approximated by the sample matrix $\Sigma_0(T)$ from single observation trajectory as $T\to\infty$. 
When $W\in\mathcal{S}_m$, the deviation norm between $R_0^x(T)$ and $\Sigma_0(T)$ goes to infinity as $T\to\infty$, and this is because the influence of the process noise will remain as the system evolves.

Next, the inference performance of OLS estimator $\widehat{W}_{o}$, the proposed $\widehat{W}_{c}$ and its correlation-based modification $\widehat{W}_{s}$ are compared under the same single observation trajectory, as demonstrated in Fig.~\ref{fig:self_error}. 
For a marginally stable NS, $\widehat{W}_{s}$ has almost the same inference performance as $\widehat{W}_{o}$. 
For an asymptotically stable NS, $\widehat{W}_{s}$ outperforms $\widehat{W}_{o}$ asymptotically but is still worse than $\widehat{W}_{c}$. 
We observe this is because the state of an asymptotically stable NS will always converge to zero, which indicates the system is mainly driven by noises regardless of the initial states. 
By joint inspection of the two cases, we note that $\widehat{W}_{c}$ applies to cases with a large observation scale, while $\widehat{W}_{s}$ applies to other situations with no worse performance than $\widehat{W}_{o}$. 
The main reason is that the statistical characteristic of observation noises will matter a lot when the observation scale is large, which is considered by $\widehat{W}_{o}$.

Finally, the non-asymptotic performance of $\widehat{W}_{c}$ and $\widehat{W}_{o}$ in Theorem \ref{th:converge-speed} is verified in Fig.~\ref{fig:convergence}. 
The upper bounds of the inference errors of two estimators are drawn in dashed lines, providing explicit expressions in terms of the observation number $T$. 
From the inset plots, we can appreciate that the proposed $\widehat{W}_{c}$ exhibits better performance than $\widehat{W}_{o}$. 
Remarkably, as $T$ increases, the inference error of $\widehat{W}_{c}$ will converge to zero while that of $\widehat{W}_{o}$ is constant.

\subsection{Comparisons with State-of-the-art Algorithms}
In this experiment, we use the case $W\in \mathcal{S}_m$, where the initial states of all agents are randomly selected from the interval $[-10,10]$, and the variance of the process and observation noise satisfy $\sigma_{\theta}^2=0.2$ and $\sigma_{\upsilon}^2=0.1$. 
Under the same initial setting, we run the dynamical process 20 times and average the following three popular evaluation metrics: the normalized mean square error (NMSE), edge identification error rate(EIER), and F-score (FS)
\begin{align}
&\operatorname{NMSE}(\widehat{W}, W)=\frac{\|\widehat{W}-W \|_{F}}{\left\|W\right\|_{F}}, \\
&\operatorname{EIER}(\widehat{W}, W)=\frac{\|W-\widehat{W}\|_{0}}{n(n-1)}, \\
&\operatorname{FS}(\widehat{W}, W)=\frac{2 \mathrm{tp}}{2 \mathrm{tp}+\mathrm{fn}+\mathrm{fp}}. 
\end{align}
Note that F-score is commonly adopted to describe the binary classification performance by computing the true-positive ($\mathrm{tp}$), false-positive ($\mathrm{fp}$) and false-negative ($\mathrm{fn}$) edge detection in estimated $\widehat{W}$. 
The value of F-score locates in $[0,1]$, where $1$ indicates perfect edge classification. 

Fig.~\ref{fig:com_test} presents the comparison results of the proposed $\widehat{W}_{c}$ with the GroupLasso algorithm in \cite{bolstad2011causal}, SPICE algorithm in \cite{venkitaraman2019learning}, and TISO algorithm in \cite{zaman2021online}. 
Fig.~\ref{test_a} and Fig.~\ref{test_b} depict the $\operatorname{NMSE}$ and $\operatorname{EIER}$ curves, respectively. 
It is clear that with the increasing observation number, the proposed estimator outperforms other algorithms in both $\operatorname{NMSE}$ and $\operatorname{EIER}$ metrics. 
Fig.~\ref{test_c} plots the F-score curve, where the proposed estimator achieves significant improvement with the observation number increasing and approaches to perfect classification. 
Notice that the $\operatorname{EIER}$ and F-score of both GroupLasso and SPICE generally do not possess an improvement with the observation number increasing. 
We conclude that this consequence may be incurred by the sparsity regularization (which is initially addressing the insufficient observation issue and turns to obtain a sparse topology), thus making the edge detection performance not well.

\subsection{Experiments on Nonlinear Cases}
In this part, we focus on the evaluation of Algorithm \ref{nonlinear-algo} in nonlinear dynamics cases, namely, inferring the binary topology structure in the NS. 
To this end, we adopt two representative cases of nonlinear model $x_{t+1}^{i}=x_{t}^{i}+ \sum\nolimits_{j = 1}^{n} \varphi_{ij}( x_{t}^{j}-x_{t}^{i} ) + \theta_t^i$, where $\varphi_{ij}$ is given by 
\begin{small}
\begin{equation}\label{eq:nonlinear_cases}
\begin{aligned}
&\text{Case 1}: \varphi_{ij}=\frac{\operatorname{sign}(a_{i j}) |x_{t}^{j}-x_{t}^{i}|(x_{t}^{j}-x_{t}^{i})}{1+\sum\nolimits_{j \in \mathcal{N}_{i}} \left(x_{t}^{j}-x_{t}^{i}\right)^{2}},  \\
&\text{Case 2}: \varphi_{ij}={a_{i j}(x_{t}^{j}-x_{t}^{i})}{(\frac{2}{1+\exp\{-(x_{t}^{j}-x_{t}^{i})\}}-1)}.
\end{aligned} 
\end{equation}
\end{small}
\!\!According to the sufficient conditions in \cite{moreau2005stability}, both the two systems will reach stable states when noise-free. 
We repeat the experiments under the the same process noise level $\sigma_{\theta}=1$, and different observation noise levels, where $\sigma_{\upsilon}$ is set as $0.1$, $0.4$, $0.7$ and $1$, respectively. 
Here the EIER metric is used to evaluate the performance of Algorithm \ref{nonlinear-algo}. 
It is clear from Fig.~\ref{final_nonlinear} that the inference accuracy generally grows with the observations, and larger observation variance will cause worse inference performance, which corresponds to the common intuition. 
As indicated in this experiment, with appropriate number of observations, the proposed algorithm can effectively infer the edge connections of NDs even with nonlinear dynamics. 

\begin{figure}[t]
\centering
\includegraphics[width=0.35\textwidth]{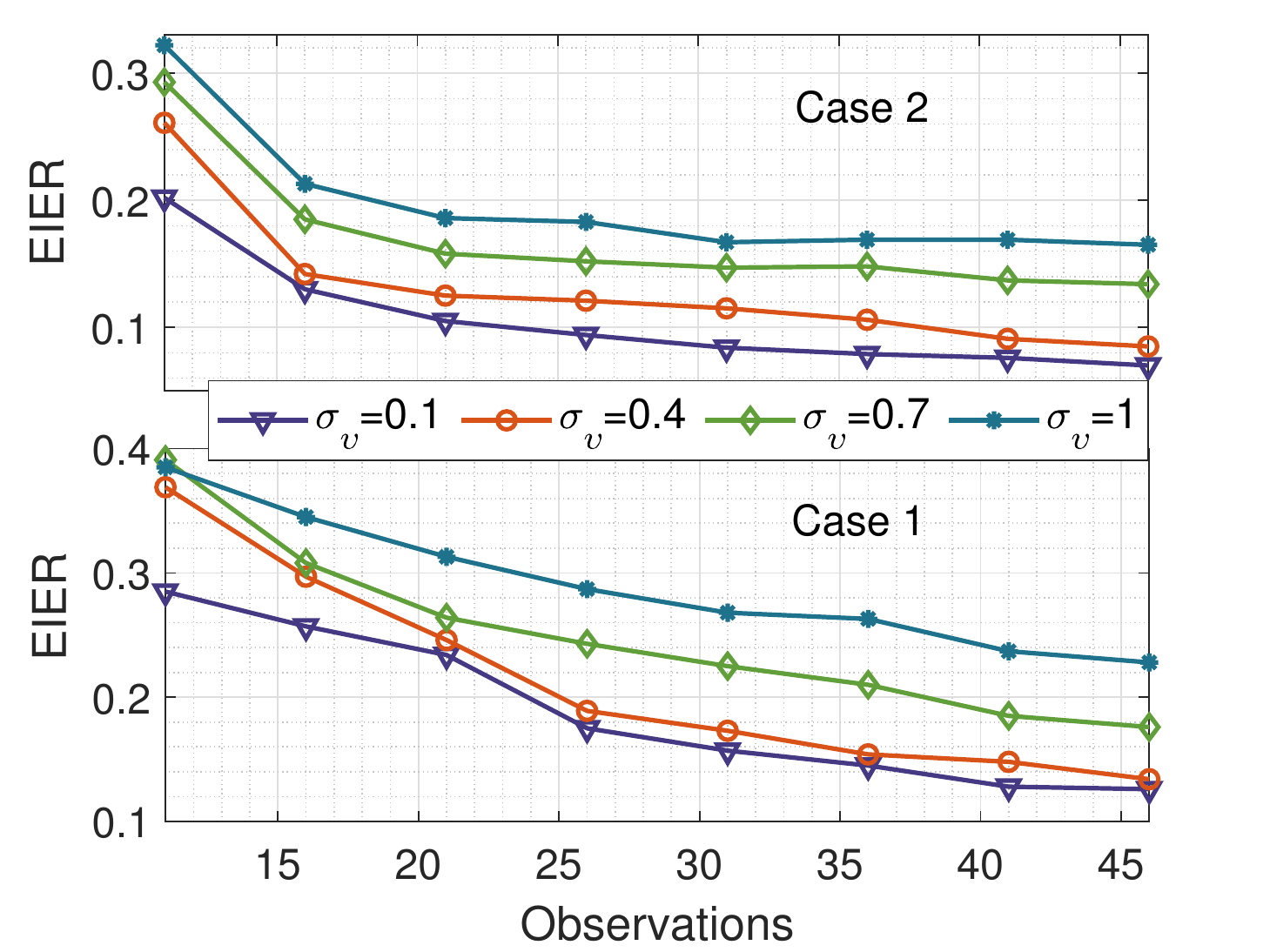}
\vspace{-5pt}
\caption{Examples of inferring the topology of NSs with nonlinear dynamics, considering the cases described in \eqref{eq:nonlinear_cases}. }
\label{final_nonlinear}
\vspace*{-13pt}
\end{figure}













\section{Conclusions}\label{conclusion}

In this paper, we investigate the directed topology inference problem of NSs in state-space representation, and characterize the non-asymptotic performance under different system stability. 
First, we proposed a causality-based estimator that allows for the presence of unknown observation noises, along with its correlation-based modification design to alleviate performance degradation when observations are few. 
By relating the proposed estimator from a single trajectory with the ideal Granger estimator from multiple trajectories, we proved their equivalence when the NS is asymptotically stable, and demonstrated the incremental characteristic of their sample matrices deviation in marginally stable cases. 
Then, we rigorously analyzed the convergence rate and accuracy of the proposed estimator by utilizing concentration measure, demonstrating that our method has superior inference performance compared with the OLS estimator. 
Besides, we provided an online version of the proposed estimator and discussed the extensions on nonlinear cases. 
Finally, extensive simulations verified our theoretical findings and showed the outperformance of the proposed estimator by comparisons.

The study of this paper provides meaningful insights into the topology inference problems, and paves the way for several interesting avenues of future research, including i) investigating a richer class of NS models, including non-stochastic input and generative switching topologies cases; 
ii) developing more novel nonlinear inference algorithms (e.g., distributed) with certain node causality and correlation as priors; 
iii) using the inference methods to topology-related applications, like anomaly detection and state prediction in NSs.

\appendix
\subsection{Proof of Lemma \ref{le:noise-state}}\label{pr:noise-state}
\begin{proof}
The proof is conducted in element-wise analysis. 
To ease notation, we denote $\Xi(T)=\frac{1}{T}\Theta_T X_T^\mathsf{T}=\sum\limits_{t = 1}^{T }\theta_{t-1}x_{t-1}^\mathsf{T}$, and the element $\Xi_{ij}(T)=\frac{1}{T}[\Theta_T X_T^\mathsf{T}]^{ij}$ is calculated by
\begin{align}
\Xi_{ij}(T)=\frac{1}{T}\sum\limits_{t = 1}^{T }\theta_{t-1}^{i} x_{t-1}^{j}.
\end{align}
Since $\theta_{t}\sim \mathcal{N}(0,\sigma^2)$, it follows that
\begin{align}
&\mathbb{E}[\Xi_{ij}(T)]= \frac{1}{T}\sum\limits_{t = 1}^{T } \mathbb{E}[\theta_{t-1}^{i}] x_{t-1}^{j}=0, \\
&\mathbb{D}[\Xi_{ij}(T)]=\sum\limits_{t = 1}^{T } ( \frac{ x_{t-1}^{j}}{T})^2 \sigma^2 \le \frac{ (|x|_{\max}^{j})^2 \sigma^2}{T},
\end{align}
where $|x|_{\max}^{j}=\max\{| x_t^j |, t=0,1,\cdots, T-1\}$.
By the famous Chebyshev inequality, given arbitrary $\epsilon >0$, we have
\begin{align}
\Pr\{ |\Xi_{ij}(T)| < \epsilon \} \! \ge \! 1- \frac{\mathbb{D}[\Xi_{ij}(T)]}{\epsilon^2} \!\ge \! 1- \frac{(|x|_{\max}^{j})^2 \sigma^2}{T\epsilon^2} .
\end{align}
Consequently, $\forall i,j\in\mathcal{V}$, $\Pr\{ |\Xi_{ij}(T)| < \epsilon \} \! \ge \! 1- \frac{(|x|_{\max}^{j})^2 \sigma^2}{T\epsilon^2} \! \ge \! 1- \frac{(|x|_{\max})^2 \sigma^2}{T\epsilon^2}$, which completes the first statement.

Next, if $|x|_{m}<\infty, \forall T\in\mathbb{N}^+$, when $T\to \infty$, it yields that
\begin{align} \label{eq:pr-con}
\mathop {\lim }\limits_{T \to \infty } \Pr\{ | \Xi_{ij}(T)| < \epsilon \} = 1.
\end{align}
Finally, in the matrix form, (\ref{eq:pr-con}) is equivalent to $ \Pr\{ \mathop {\lim }\limits_{T \to \infty } \Xi(T) \!=\!\bm{0} \}\!=\!1$.
The proof is completed.
\end{proof}

\subsection{Proof of Theorem \ref{th:causality_estimator}}\label{pr:causality_estimator}
\begin{proof}
Substitute (\ref{eq:two-observation}) into $\Sigma_1(T)$ and it follows that
\begin{align}
&\Sigma_1(T)=\frac{1}{T}(Y_T^+) (Y_T^-)^\mathsf{T}=\frac{1}{T}\sum\limits_{t = 1}^{T} y_{t} y_{t-1}^\mathsf{T} \nonumber \\
=&\frac{W}{T} \sum\limits_{t = 1}^{T }y_{t-1}y_{t-1}^\mathsf{T} \! + \!\frac{1}{T} \sum\limits_{t = 1}^{T }\left(\theta_{t-1} \!+\! \upsilon_{t}-W \upsilon_{t-1}\right)y_{t-1}^\mathsf{T}.
\end{align}

Note that when $W\in \mathcal{S}_a$, $\mathop {\lim }\limits_{t \to \infty } \|y_{t}\|< \infty$ holds almost surely. 
Since $\theta_{t-1}$ and $\upsilon_{t}$ are independent of $y_t$, applying Lemma \ref{le:noise-state} on $\Sigma_1(T)$, it yields that
\begin{equation}
\mathop {\lim }\limits_{T \to \infty } \frac{1}{T} \sum\limits_{t = 1}^{T }\theta_{t-1}y_{t-1}^\mathsf{T}=\bm{0},~\mathop {\lim }\limits_{T \to \infty } \frac{1}{T} \sum\limits_{t = 1}^{T }\upsilon_{t}y_{t-1}^\mathsf{T}=\bm{0}.
\end{equation}
Recalling $\upsilon_{t-1}$ is independent of $x_{t-1}$, it follows that
\begin{small}
\begin{equation}\nonumber
\mathop {\lim }\limits_{T \to \infty } \frac{1}{T} \sum\limits_{t = 1}^{T }\upsilon_{t-1}y_{t-1}^\mathsf{T}=\mathop {\lim }\limits_{T \to \infty } \frac{1}{T} \sum\limits_{t = 1}^{T }\upsilon_{t-1}(x_{t-1}+\upsilon_{t-1})^\mathsf{T}=\sigma_{\upsilon}^2I.
\end{equation}
\end{small}
\!\!Then, one infers that
\begin{equation}
\mathop {\lim }\limits_{T \to \infty }\Sigma_1(T)= W\left(\mathop {\lim }\limits_{T \to \infty }\Sigma_0(T) - \sigma_{\upsilon}^2I \right).
\end{equation}
Hence, the proof is completed.
\end{proof}

\subsection{Proof of Theorem \ref{th:equivalent1}}\label{pr:equivalent1}
\begin{proof}
Without losing generality, we first consider $\Sigma_0(T)$.
Substituting the expanded form (\ref{eq:expand_form}) of $y_t$ into $y_{t}y_{t}^\mathsf{T}$, one obtains
\begin{align}\label{eq:yy}
y_{t}y_{t}^{\mathsf{T}}=(W^{t}x_{0}+\eta_t) (W^{t}x_{0}+\eta_t )^\mathsf{T},
\end{align}
where $\eta_t=\sum\limits_{m = 1}^{t} W^{m-1} \theta_{t-m}+\upsilon_t$ and $\eta_0=\upsilon_0$.
Then, $y_{t}y_{t}^{\mathsf{T}}$ is expanded as
\begin{align}\label{eq:yy2}
y_{t}y_{t}^{\mathsf{T}}\!=\!\underbrace{ W^{t}x_{0}x_{0}^{\mathsf{T}}(W^{t})^\mathsf{T} }_{Q_1^t} \!+\! \underbrace{W^{t}x_{0} \eta_t^\mathsf{T} }_{Q_2^t} \!+\! \underbrace{\eta_t x_{0}^\mathsf{T} (W^{t})^\mathsf{T} }_{Q_3^t} \!+\! \underbrace{\eta_t \eta_t^\mathsf{T}}_{Q_4^t}.
\end{align}
Based on (\ref{eq:yy2}), the proof is equivalent to separately find the average of the summation of each part in (\ref{eq:yy2}) along the observation horizon $T$. 

First, consider taking the average of all $\{Q_1^t\}_{t=0}^{T-1}$.
Note that when $W\in \mathcal{S}_a \cup \mathcal{S}_m$, $W^{\infty}x_0$ converge to a constant vector.
Therefore, it yields that
\begin{align}\label{eq:q1-sum}
\mathop {\lim }\limits_{T \to \infty }\frac{1}{T} \sum\limits_{t=0 }^{T-1 }Q_1^t= &\mathop {\lim }\limits_{T \to \infty } \frac{1}{T} \sum\limits_{t=0 }^{T-1 }W^t x_{0}x_{0}^{\mathsf{T}} (W^t)^\mathsf{T} \nonumber \\
=&W^{\infty} x_{0}x_{0}^{\mathsf{T}} (W^{\infty})^\mathsf{T}.
\end{align}

Next, consider $\frac{1}{T} \sum\limits_{t=0}^{T-1 }Q_2^t \!=\! \frac{1}{T} \sum\limits_{t=0}^{T-1} W^{t}x_{0} \eta_t^\mathsf{T}$.
Since $\eta_t$ is a typical linear combination of Gaussian noises $\{\theta_m\}_{m=0}^{t-1}$ and $\upsilon_t$, and is independent of $W^{t}x_{0}$,
by Lemma \ref{le:noise-state}, one infers that
\begin{align}\label{eq:convere-dd}
\mathop {\lim }\limits_{T \to \infty }\frac{1}{T} \sum\limits_{t=0}^{T-1 }Q_2^t= \mathop {\lim }\limits_{T \to \infty } \frac{1}{T}\sum\limits_{t = 0}^{T-1} (W^{t}x_{0}) \eta_t^\mathsf{T}=\bm{0}.
\end{align}
The average of all $\{Q_3^t\}_{t=0}^{T-1}$ is likewise, i.e.,
\begin{align}\label{eq:convere-dd2}
\mathop {\lim }\limits_{T \to \infty }\frac{1}{T} \sum\limits_{t=0}^{T-1 }Q_3^t=\bm{0}.
\end{align}

Then, focus on the calculation of $\frac{1}{T} \sum\limits_{t=0}^{T-1 }Q_4^t$.
Since $\mathbb{E}( \theta_{t}\theta_{t}^\mathsf{T} )=\sigma_{\theta}^2 I$ and $\mathbb{E}( \upsilon_{t}\upsilon_{t}^\mathsf{T} )=\sigma_{\upsilon}^2 I$, one can divide $Q_4^t$ as
\begin{align} \label{eq:q4}
Q_4^t =& \sum\limits_{m = 0}^{t} \left( W^{m-1} \theta_{t-m} (\upsilon_t + \sum\limits_{ m_1\neq m}^{t} \theta_{t-m_1}^\mathsf{T} (W^{m_1-1})^\mathsf{T} ) \right) \nonumber \\
&+\upsilon_{t} \left(\sum\limits_{m_1 = 0}^{t} \theta_{t-m_1}^\mathsf{T} (W^{m_1-1})^\mathsf{T} \right) \nonumber \\
&+ \sum\limits_{m = 0}^{t} W^{m-1} \theta_{t-m} \theta_{t-m}^\mathsf{T} (W^{m-1})^\mathsf{T} + \upsilon_{t}\upsilon_{t}^\mathsf{T}.
\end{align}
Consider the first term in $Q_4^t$.
For simple expression, define
\begin{equation}
\theta_t^a(m)= W^{m-1} \theta_{t-m},~\theta_{t}^b(m)=\sum\limits_{m_1\neq m}^{t} W^{m_1-1}\theta_{t-m_1}.
\end{equation}
As $W\in \mathcal{S}_a$, one can infer that
$\mathop {\lim }\limits_{t \to \infty } \|\theta_t^b(m)\|<\infty$.
Therefore, by the famous Lebesgue's dominated convergence theorem and Lemma \ref{le:noise-state}, it follows that with probability one
\begin{align}
&\mathop {\lim }\limits_{T \to \infty } \frac{1}{T} \sum\limits_{t=1}^{T} \sum\limits_{m = 0}^{t} \theta_t^a(m) (\theta_{t}^b(m))^\mathsf{T} \nonumber \\
=& \sum\limits_{t=0}^{\infty} \mathop {\lim }\limits_{T \to \infty } \frac{\sum\limits_{m = 0}^{t} \theta_t^a(m) (\theta_{t}^b(m))^\mathsf{T} }{T}=\bm{0}.
\end{align}
Likewise, for the second term in (\ref{eq:q4}), it also holds that
\begin{equation}
\mathop {\lim }\limits_{T \to \infty } \frac{1}{T} \sum\limits_{t=0}^{T-1} \upsilon_{t} \left(\sum\limits_{m_1 = 0}^{t} \theta_{t-m_1}^\mathsf{T} (W^{m_1 -1})^\mathsf{T} \right) = \bm{0}.
\end{equation}

As for the last two parts in (\ref{eq:q4}), recalling $\mathbb{D}[\theta_t]\!=\!\sigma_{\theta}^2I$ and $\mathbb{D}[\upsilon_t]\!=\!\sigma_{\upsilon}^2I$, and one infers that
\begin{align}\label{eq:q4-sum}
&\mathop {\lim }\limits_{T \to \infty } \frac{ \sum\limits_{t = 0}^{T-1} \sum\limits_{m = 0}^{t} W^{m-1} \theta_{t-m} \theta_{t-m}^\mathsf{T} (W^{m-1})^\mathsf{T} + \sum\limits_{t = 0}^{T-1} \upsilon_{t}\upsilon_{t}^\mathsf{T} }{T} \nonumber \\
=& \sum\limits_{t = 0}^{\infty} W^t \mathop {\lim }\limits_{T \to \infty } \left ( { \sum\limits_{m = 0}^{T-t} \theta_{T-m} \theta_{T-m}^\mathsf{T} }/T \right ) (W^t)^\mathsf{T} \nonumber + \sigma_{\upsilon}^2 I \\
=& \sigma_{\theta}^2 \sum\limits_{t = 0}^{\infty} W^t (W^t)^\mathsf{T} + \sigma_{\upsilon}^2 I = \mathop {\lim }\limits_{T \to \infty }\frac{1}{T} \sum\limits_{t=0}^{T-1}Q_4^t.
\end{align}

Finally, taking (\ref{eq:q1-sum})-(\ref{eq:convere-dd2}) and (\ref{eq:q4-sum}) into $\mathop {\lim }\limits_{T \to \infty } \Sigma_0(T)$ implies
\begin{align}\label{}
\mathop {\lim }\limits_{T \to \infty } \Sigma_0(T) & = \mathop {\lim }\limits_{T \to \infty } \frac{1}{T} \sum\limits_{t = 0}^{T-1} (Q_1^t+ Q_2^t + Q_3^t +Q_4^t)\nonumber \\
&= \mathop {\lim }\limits_{T \to \infty } \frac{1}{T} \sum\limits_{t = 0}^{T-1} (Q_1^t + Q_4^t)\nonumber \nonumber \\
&=W^{\infty} x_{0}x_{0}^{\mathsf{T}} (W^{\infty})^\mathsf{T}+\sigma_{\theta}^2 \sum\limits_{t = 0}^{\infty} W^t (W^t)^\mathsf{T} + \sigma_{\upsilon}^2 I \nonumber \\
&=R_0^x(\infty) + \sigma_{\upsilon}^2 I.
\end{align}
The proof of $\Sigma_1(\infty)=R_1^x(\infty)+\sigma_{\upsilon}^2 W$ is likewise and omitted here.
The proof is completed.
\end{proof}

\subsection{Proof of Theorem \ref{th:relationship}}\label{pr:relationship}
\begin{proof}
To prove this theorem, we first present the characteristic of the solution of $\textbf{P}_\textbf{2}$. 
Recall that the OLS version for solving $W_i$ is formulated as 
\begin{equation}\label{eq:wi_OLS}
\min_{W_{i} \in \mathbb{R}^{1\times n}} \left\| Z_T W_{i}^{\mathsf{T}} -\bm{b}_{T}^{i}\right\|^2 .
\end{equation}
Note that both the coefficient matrix $Y_{T}^-$ and the observation vector $\bm{b}_{T}^{i}$ are corrupted by noises. 
Considering this point, in the optimization community, the formulation of $\textbf{P}_\textbf{2}$ can be interpreted as total least squares (TLS) problem \cite{golub1980analysis}, which is a weighted-squares version of \eqref{eq:wi_OLS}. 
Define the the augmented matrix $B_T=[Z_T,\bm{b}_{T}^{i}]\in\mathbb{R}^{T \times (n+1)}$, 
and the analytic solution of $\textbf{P}_\textbf{2}$ is given by (see Section 2.3 in \cite{markovsky2007overview})
\begin{equation}
\widehat{W}_{i,\operatorname{TLS}}^{\mathsf{T}}=\left( Z_T Z_T^{\mathsf{T}} -\rho_{\min}^{2}(B_T) I\right)^{-1} Z_T \bm{b}_{T}^{i},
\end{equation}
where $\rho_{\min}(B_T)$ is the smallest singular value of $B_T$. 
Following this, all we need is to prove $\Pr\{\mathop {\lim }\limits_{T \to \infty } \rho_{\min}^{2}(B_T) \!=\! \sigma_v^2\}=1$.

Considering the correlation of $B_T$, one has 
\begin{align}\label{eq:correlation_B}
B_T^{\mathsf{T}}B_T =&
\begin{bmatrix} Z_T^{\mathsf{T}} \\ (\bm{b}_{T}^{i})^{\mathsf{T}} \end{bmatrix} 
\begin{bmatrix} Z_T & \bm{b}_{T}^{i} \end{bmatrix}
=\begin{bmatrix} 
Z_T^{\mathsf{T}} {Z_T}  & Z_T^{\mathsf{T}} \bm{b}_{T}^{i} \\
(\bm{b}_{T}^{i})^{\mathsf{T}} {Z_T} & (\bm{b}_{T}^{i})^{\mathsf{T}} \bm{b}_{T}^{i}
\end{bmatrix} \nonumber \\
=&\begin{bmatrix} 
\frac{1}{T} Y_T^- (Y_T^-)^\mathsf{T}  &  \frac{1}{T}\sum\limits_{t =0}^{T-1} y_{t} y_{t+1}^{i} \\
\frac{1}{T}\sum\limits_{t =0}^{T-1} y_{t+1}^{i} y_{t}^\mathsf{T} & \frac{1}{T}\sum\limits_{t =1}^{T} (y_{t}^{i})^2
\end{bmatrix} \nonumber\\
=& \underbrace{ \begin{bmatrix} 
\frac{1}{T} X_T^-(X_T^-)^\mathsf{T}  & \frac{1}{T}\sum\limits_{t =0}^{T-1} x_{t} x_{t+1}^{i} \\
\frac{1}{T}\sum\limits_{t =0}^{T-1} x_{t+1}^{i} x_{t}^\mathsf{T} & \frac{1}{T}\sum\limits_{t =1}^{T} (x_{t}^{i})^2
\end{bmatrix}  }_{M_T^a \in\mathbb{R}^{(n+1)\times (n+1)}} \nonumber\\
&+\underbrace{ \begin{bmatrix} 
\frac{1}{T}\sum\nolimits_{t =0}^{T-1} \upsilon_{t} \upsilon_{t}^\mathsf{T}  & \frac{1}{T}\sum\limits_{t =0}^{T-1} \upsilon_{t} \upsilon_{t+1}^{i}\\
\frac{1}{T}\sum\limits_{t =0}^{T-1} \upsilon_{t+1}^{i} \upsilon_{t}^\mathsf{T} & \frac{1}{T} \sum\nolimits_{t =1}^{T} (\upsilon_{t}^{i})^2
\end{bmatrix}  }_{M_T^b \in\mathbb{R}^{(n+1)\times (n+1)}}
.
\end{align}
Note that for the term $M_T^a \in\mathbb{R}^{(n+1)\times (n+1)}$ in \eqref{eq:correlation_B}, its last row is identical with $i$-th row, thus yielding that 
\begin{equation}\label{eq:zero_eig}
\operatorname{Rank}(M_T^a)=n ~\Rightarrow ~\lambda_{\min}(M_T^a)=0, 
\end{equation}
where $\lambda_{\min}(\cdot)$ denotes the smallest eigenvalue of a square matrix. 
For the term $M_T^b$, recall that $\upsilon_t$ is i.i.d. Gaussian noises and subject to ${N}(0,\sigma^2_{\upsilon} I)$. 
When $T\to\infty$, it follows that with probability one  
\begin{align}
\mathop {\lim }\limits_{T \to \infty } M_T^b & = 
\begin{bmatrix} 
\mathop {\lim }\limits_{T \to \infty } \frac{1}{T}\sum\nolimits_{t =0}^{T-1} \upsilon_{t} \upsilon_{t}^\mathsf{T}  &  \mathop {\lim }\limits_{T \to \infty } \frac{1}{T}\sum\limits_{t =0}^{T-1} \upsilon_{t} \upsilon_{t+1}^{i}\\
\mathop {\lim }\limits_{T \to \infty } \frac{1}{T}\sum\limits_{t =0}^{T-1} \upsilon_{t+1}^{i} \upsilon_{t}^\mathsf{T} & \mathop {\lim }\limits_{T \to \infty } \frac{1}{T} \sum\nolimits_{t =1}^{T}(\upsilon_{t}^{i})^2 
\end{bmatrix} \nonumber  \\
&= \sigma_{\upsilon}^2 I, 
\end{align}
where the fact that $\upsilon_{t}$ and $\upsilon_{t+1}^{i}$ are independent of each other is adopted, and thus $\mathop {\lim }\limits_{T \to \infty } \frac{1}{T}\sum\limits_{t =0}^{T-1} \upsilon_{t+1}^{i} \upsilon_{t}^\mathsf{T}=\bm{0}$ holds. 

Finally, utilizing \eqref{eq:zero_eig} and the property that $\lambda_{\min}(M+\sigma^2 I)=\lambda_{\min}(M)+\sigma^2$ ($M$ is an arbitrary square matrix), the minimal eigenvalue of $B_T^{\mathsf{T}}B_T$ satisfies 
\begin{equation}
\mathop {\lim }\limits_{T \to \infty } \lambda_{\min}(B_T^{\mathsf{T}}B_T) \!=\! \mathop {\lim }\limits_{T \to \infty } (\lambda_{\min}(M_T^a) \!+\! \lambda_{\min}(M_T^b))\!=\!\sigma_{\upsilon}^2, 
\end{equation}
which is also the square of the minimal singular value of $B_T$, i.e., $\mathop {\lim }\limits_{T \to \infty } \rho_{\min}^{2}(B_T) \!=\! \sigma_{\upsilon}^2$. 
The proof is completed. 
\end{proof}

\subsection{Proof of Theorem \ref{th:equivalent-nonstable}}\label{pr:equivalent-nonstable}
\begin{proof}
We proceed this proof based on the analysis in Theorem \ref{th:equivalent1}. 
The key point is to reveal the growing characteristic of the deviation norm in terms of $T$.  
Recall $\eta_t=\sum\limits_{m = 1}^{t} W^{m-1} \theta_{t-m}+\upsilon_t$ ($\eta_0=\upsilon_0$), and the deviation matrix $(\Sigma_0(T) -R_0^x(T)-\sigma_{\upsilon}^2 I)$ is expanded as 
\begin{small}
\begin{align}\label{eq:original_deviation}
 \!\!&\Sigma_0(T) -R_0^x(T)-\sigma_{\upsilon}^2 I \nonumber  \\
\!\!\!= &\frac{1}{T} (\sum\limits_{t = 0}^{T-1} (W^{t}x_{0}) \eta_t^\mathsf{T} \!\!+\! \!  \sum\limits_{t = 0}^{T-1} \eta_t x_{0}^\mathsf{T} (W^{t})^\mathsf{T} \! \!+\!\!  \sum\limits_{t=0}^{T-1} \eta_t \eta_t^\mathsf{T} )  \!-\! \sigma_{\theta}^2 \sum\limits_{t = 0}^{T-1} W^t (W^t)^\mathsf{T} \nonumber  \\
& + \frac{1}{T} \sum\limits_{t=0 }^{T-1 }W^t x_{0}x_{0}^{\mathsf{T}} (W^t)^\mathsf{T} - W^{T-1} x_{0}x_{0}^{\mathsf{T}} (W^{T-1})^\mathsf{T}  -\sigma_{\upsilon}^2 I .
\end{align}
\end{small}
\!\!Due to $W\in \mathcal{S}_m$, $\mathop {\lim } \limits_{m\to \infty } W^m$ exists, i.e., $|[W^m]^{ij}|<\infty,~\forall m \in \mathbb{R}^{+}$. 
Therefore, the spectral norm of the last three terms in (\ref{eq:original_deviation}) is bounded. 
Then, the major focus is laid on the first three terms. 
Considering an element-wise analysis of $W^{m-1} \theta_{t-m}$, the noise variance in every dimension is given by 
\begin{equation}
\mathbb{D}( [W^{m-1} \theta_{t-m}]^i )= \sigma_{\theta}^2 \sum\limits_{j = 1}^{n} ([W^{m-1}]^{ij} )^2, \forall i \in \mathcal{V}. 
\end{equation}
Since $[W^m]^{ij}$ is strictly bounded, there exists a $n$-dimension state vector $\tilde x$ and a Gaussian noise $\tilde \theta \sim {N}(0,\tilde \sigma^2 I)$ such that 
\begin{align}\label{eq:state_bound11}
\!\! \|W^{m}x_0\| \! \le \!\| \tilde x\| \!<\!\infty,~\mathbb{D}( [W^{m-1} \theta_{t-m}]^i ) \!\le\! \mathbb{D}( \tilde\theta^{[i]} )\!=\!\tilde \sigma^2.
\end{align}
Based on (\ref{eq:state_bound11}), we define a revised version of $\eta_t$ by $\tilde \eta_t=\sum\limits_{m = 1}^{t} \tilde\theta_t +\upsilon_t$ ($\tilde \eta_0=\upsilon_0$), and the alternative of the first three terms in (\ref{eq:original_deviation}) is given by
\begin{align}
E_{\theta}= \frac{1}{T} ( \underbrace{ \tilde x \sum\limits_{t = 0}^{T-1} \tilde\eta_t^\mathsf{T} }_{J_1(T)} +   \underbrace{  \sum\limits_{t = 0}^{T-1} \tilde \eta_t \tilde x^\mathsf{T} }_{J_2(T)}  + \underbrace{  \sum\limits_{t=1}^{T} \tilde\eta_t \tilde\eta_{t}^\mathsf{T} }_{J_3(T)}   ). 
\end{align}
In the sequel, we turn to analyze the the asymptotic performance of $\|E_{\theta}\|$ to demonstrate that of $\| \Sigma_0(T) -R_0^x(T) - \sigma_{\upsilon}^2 I \|$.

First, look at the each entry in $J_1^{[ij]}(T)=\tilde x^{[i]} \sum\limits_{t = 0}^{T-1} \tilde\eta_t^{[j]}$, which satisfies  
\begin{align}
\mathbb{E}\{J_1^{[ij]}(T)\}=0,~\mathbb{D}\{J_1^{[ij]}(T)\}=(\tilde x^{[i]})^2 \tilde\sigma^2 T .
\end{align}
By the Chebyshev inequality, given $0<\delta<1$, one has 
\begin{align}
\Pr\{ |J_1^{[ij]} (T)| \le \sqrt{\frac{T}{\delta}} \tilde \sigma |\tilde x^{[i]}|  \} \ge 1-\delta.
\end{align}
Therefore, it yields that at least with probability $1-\delta$,
\begin{align}\label{eq:t_bound1}
|J_1^{[ij]}(T)/T| \le  \sqrt{\frac{1}{T \delta}} \tilde \sigma\tilde x^{m}, \forall i,j\in\mathcal{V},
\end{align}
where $\tilde x^{m}=\max\{|\tilde x^{[i]}|, i\in\mathcal{V}\}$. 
Note that $J_1(T)=J_2^\mathsf{T}(T)$, and thus the bound in (\ref{eq:t_bound1}) also applies to $J_2^{[ij]}(T)$. 

Next, consider the entries in $J_3(T)$. 
Since the autocorrelation of $\{ \tilde\theta_t \}_{t=1}^{T} $ and $\{ \tilde\upsilon_t \}_{t=1}^{T} $ are involved, $J_3(T)$ can be further expanded as 
\begin{align}
J_3(T) =&  \sum\limits_{t = 0}^{T-1} ( \sum\limits_{t_1 = 0}^{t}\sum\limits_{t_2 = 0,\atop t_2\neq t_1}^{t} \tilde\theta_{t_1} \tilde\theta_{t_2}^\mathsf{T} + \sum\limits_{t_1 = 0}^{t} \tilde\theta_{t_1} \tilde\theta_{t_1}^\mathsf{T}   \nonumber \\
&  + \sum\limits_{t_1 = 0}^{t} \tilde\theta_{t_1} \upsilon_{t_1}^\mathsf{T} +  \sum\limits_{t_1 = 0}^{t} \tilde\upsilon_{t_1} \theta_{t_1}^\mathsf{T} +\upsilon_{t}\upsilon_{t}^\mathsf{T} ).
\end{align}
Recall that the product of two independent Gaussian variable also subjects to Gaussian distribution, thus it follows that 
\begin{align}
\!\!\left \{
\begin{aligned}
&\mathbb{E}\{(\tilde\theta_{t_1} \tilde\theta_{t_2}^\mathsf{T})^{[ij]}\}=0, \\
&\mathbb{D}\{(\tilde\theta_{t_1} \tilde\theta_{t_2}^\mathsf{T})^{[ij]}\}=\frac{\tilde \sigma^2}{2},
\end{aligned}
\right.
\left \{
\begin{aligned}
&\mathbb{E}\{(\tilde\theta_{t_1} \tilde\upsilon_{t_1}^\mathsf{T})^{[ij]}\}=0, \\
&\mathbb{D}\{(\tilde\theta_{t_1} \tilde\upsilon_{t_1}^\mathsf{T})^{[ij]}\}=\frac{\tilde\sigma^2 \sigma_{\upsilon}^2 }{\tilde\sigma^2 + \sigma_{\upsilon}^2},
\end{aligned}
\right.
\end{align}
where $t_1\neq t_2$. 
Applying the Chebyshev inequality again, for each entry in $J_3(T)$, one has with probability at least $1-\delta$
\begin{align}\label{eq:upup}
|J_3^{[ij]}(T)| \le & \sqrt{\frac{\sum\limits_{t = 1}^{T-1}t(t+1)} {2\delta} } \tilde \sigma + 2 \sqrt{\frac{\sum\limits_{t = 1}^{T}t } {\delta} } \sqrt{\frac{\tilde\sigma^2 \sigma_{\upsilon}^2 }{\tilde\sigma^2 + \sigma_{\upsilon}^2}} \nonumber \\
& + (\sum\limits_{t =1}^{T} t)\frac{\tilde\sigma^2}{\delta}  + \frac{T \sigma_{\upsilon}^2 }{\delta}. 
\end{align}
Then, divide $T$ into ${J_3}(T)$ and do series summation, yielding
\begin{align}\label{eq:speed}
|\frac{J_3^{[ij]}(T)}{T}|\le&  \sqrt{\frac{2 T^{3}-3 T^{2}+T}{12 T^{2} \delta}} \tilde\sigma+\sqrt{\frac{2 T^{2}-T}{T^{2} \delta}}\sqrt{\frac{\tilde\sigma^2 \sigma_{\upsilon}^2 }{\tilde\sigma^2 + \sigma_{\upsilon}^2}} \nonumber \\
 &+\frac{T+1}{2 \delta} \tilde{\sigma}^{2} + \frac{\sigma_{\upsilon}^2 }{\delta} \triangleq \bar J(T). 
\end{align}
Finally, utilizing the inequality $\|B \| \leq\|B \|_{F} \leq \sum\limits_{i,j\in\mathcal{V}} \| B^{[ij]}\|$ ($B\in\mathbb{R}^{n \times n}$), it is induced that with probability at least $1-\delta$ 
\begin{align}
\|\frac{J_3(T)}{T}\|\le&  n^2 \bar J(T) \sim\bm{O}( \frac{T}{\delta} ). 
\end{align}

Note that $\| R_0^x(T) \|=\| \sigma_{\theta}^2 \sum\nolimits_{t = 0}^{T-1} W^t (W^t)^\mathsf{T} \|\le T \sigma_{\theta}^2$, then there exists a possibility that the part with factor $T$ in $\frac{J_3(T)}{T}$ can be offset with $R_0^x(T)+ \sigma_{\upsilon}^2 I$, i.e., 
\begin{small}
\begin{equation}\label{eq:possibility}
\|R_0^x(T)+ \sigma_{\upsilon}^2 I- \frac{ \sum\limits_{t = 0}^{T-1} \sum\limits_{m = 0}^{t} W^{m-1} \theta_{t-m} \theta_{t-m}^\mathsf{T} (W^{m-1})^\mathsf{T} }{T} \| 
\stackrel{T\to\infty}{\longrightarrow} \bm{O}(1). 
\end{equation}
\end{small}
\!\!However, even if the situation in (\ref{eq:possibility}) happens, by (\ref{eq:speed}) one can infer that $\| \Sigma_0(\infty) -R_0^x(\infty) - \sigma_{\upsilon}^2 I \| \sim\bm{O}(\sqrt{T})$ holds, indicating the deviation still goes to infinity with the increase of $T$. 
Hence, the spectral norm of the deviation (\ref{eq:original_deviation}) at least satisfies $\bm{O}(\sqrt{T})$, which completes the proof. 
\end{proof}

\subsection{Proof of Theorem \ref{th:asymptotic-performance}}\label{pr:asymptotic-performance}
\begin{proof}
The proof is rather similar to that of Theorem \ref{th:equivalent1}. 
The key idea is to prove the upper bounds of the estimation error, by leveraging the concentration measure in Gaussian space. 
Since $\widehat{W}_o \!= \!Y_T^+ (Y_T^-)^\mathsf{T} (Y_T^- (Y_T^-)^\mathsf{T})^{-1}$, one obtains the estimation error matrix $E_W$ by 
\begin{equation}
E_W=\widehat{W}_o-W=(\Theta_{T}+\Upsilon_{T}^{+}  - W\Upsilon_{T}^{-}) (Y_{T}^{-})^\mathsf{T} \Sigma_{T}^{-1}. 
\end{equation} 
Due to $\| W \| \!\le \!1$ and $Y_{T}^{-}\!=\! X_{T}^{-}\!+\!\Upsilon_{T}^{-} $,  $\| E_W \|$ is bounded by 
\begin{align}\label{eq:upperbounds}
\| E_W \| \le& \| \Theta_{T} (Y_{T}^{-})^\mathsf{T} \Sigma_{T}^{-1} \| + \|\Upsilon_{T}^{+} (Y_{T}^{-})^\mathsf{T} \Sigma_{T}^{-1}\| \nonumber \\
&+ \| \Upsilon_{T}^{-} (X_{T}^{-})^\mathsf{T} \Sigma_{T}^{-1} \| + \| \Upsilon_{T}^{-} (\Upsilon_{T}^{-})^\mathsf{T} \Sigma_{T}^{-1} \|.
\end{align}
Then, the proof is turned to bound each term of the right-hand side in (\ref{eq:upperbounds}) individually. 

\begin{itemize}
\item \textit{Part 1: Upper Bounding $\| \Theta_{T} (Y_{T}^{-})^\mathsf{T} \Sigma_{T}^{-1} \|$. }
\end{itemize}

Utilizing $0 \prec V_{dn} \preceq \Sigma_{T} \preceq \Sigma_{T} + V_{dn}  \preceq  V_{up}$, it yields that 
\begin{align}
& \Sigma_{T}+V_{dn} \preceq 2 \Sigma_{T} \Longrightarrow \left(\Sigma_{T}+V_{dn}\right)^{-1} \succeq \Sigma_{T}^{-1}/2 \nonumber  \\
&  ~\Longrightarrow  \| \Sigma_{T}^{-\frac{1}{2}} \| \le \sqrt{2} \|\left(\Sigma_{T}+V_{d n}\right)^{-\frac{1}{2}} \| . 
\end{align}
Notice that $\{\theta_{t},y_{t}\}$\! are mutually independent like $\{\upsilon_{t+1},y_{t}\}$. 
Therefore, considering the first term in RHS of \eqref{eq:upperbounds}, $\| \Theta_{T} (Y_{T}^{-})^\mathsf{T} \Sigma_{T}^{-1}\| \le \|\Sigma_{T}^{-\frac{1}{2}}\| \|\Theta_{T} (Y_{T}^{-})^\mathsf{T} \Sigma_{T}^{-\frac{1}{2}} \|$, 
one can directly apply Lemma \ref{le:lsm-bound1} and has probability at least $1-\delta$ that 
\begin{small}
\begin{align}\label{eq:yy_up}
& \| \Theta_{T} (Y_{T}^{-})^\mathsf{T} \Sigma_{T}^{-1}\| \le \|\Sigma_{T}^{-\frac{1}{2}}\| \|\Theta_{T} (Y_{T}^{-})^\mathsf{T} \Sigma_{T}^{-\frac{1}{2}} \| \nonumber \\ 
&\le \sqrt{2} \|\Sigma_{T}^{-\frac{1}{2}}\| \| \Theta_{T} (Y_{T}^{-})^\mathsf{T} \left(\Sigma_{T}+V_{d n}\right)^{-\frac{1}{2}} \|   \nonumber \\  
&\le \frac{4}{\sqrt{\lambda_{\min}(V_{dn})}} \sqrt{n \log (\frac{5 \operatorname{det}\left(  V_{up} V_{dn}^{-1}+ I  \right)^{ \frac{1}{2n}}  } {\delta^{1 / n}} )},
\end{align}
\end{small}
\!\!where the facts $\|\Sigma_{T}^{-\frac{1}{2}}\|\le \frac{1}{\sqrt{\lambda_{\min}(V_{dn})}}$ and $\Sigma_{T} + V_{dn}  \preceq  V_{up}$ are applied in the last inequality of \eqref{eq:yy_up}. 

\begin{itemize}
\item \textit{Part 2: Upper Bounding $\| \Upsilon_{T}^{+} (Y_{T}^{-})^\mathsf{T} \Sigma_{T}^{-1} \|$. }
\end{itemize}

Since $\{\upsilon_{t+1}, y_{t}\}$ are independent with other, thus the same upper of (\ref{eq:yy_up}) also applies to $\| \Upsilon_{T}^{+} (Y_{T}^{-})^\mathsf{T} \Sigma_{T}^{-1} \|$. 

\begin{itemize}
\item \textit{Part 3: Upper Bounding $\| \Upsilon_{T}^{-} (X_{T}^{-})^\mathsf{T} \Sigma_{T}^{-1} \|$. }
\end{itemize}

Notice that $X_{T}^{-}= Y_{T}^{-}-\Upsilon_{T}^{-}$ and $\{\upsilon_{t}, x_{t}\}$ are independent with other. 
Following this, the matrices $V_{up}$ and $V_{dn}$ can be determined such that $V_{dn} \preceq X_{T}^{-} (X_{T}^{-} )^\mathsf{T} \preceq  V_{up}$. 
Similar to the Part 1, an upper bound of $\| \Upsilon_{T}^{-} (X_{T}^{-})^\mathsf{T} \Sigma_{T}^{-1} \|$ is given by 
\begin{small}
\begin{align}
\| \Upsilon_{T}^{-} (X_{T}^{-})^\mathsf{T} \Sigma_{T}^{-1} \| \!\le\! \frac{4}{\sqrt{\lambda_{\min}(V_{dn})}} \sqrt{n \log (\frac{5 \operatorname{det} (  V_{up} V_{dn}^{-1} \!+\! I )^{ \frac{1}{2n}}  } {\delta^{1 / n}})}. 
\end{align}
\end{small}


\begin{itemize}
\item \textit{Part 4: Upper Bounding $\| \Upsilon_{T}^{-} (\Upsilon_{T}^{-})^\mathsf{T} \Sigma_{T}^{-1} \|$. }
\end{itemize}

Applying Lemma \ref{le:single-value} to $\Upsilon_{T}^{-} (\Upsilon_{T}^{-})^\mathsf{T}$, and one obtains that with probability at least $1- 2 \exp \left(-r^{2} / 2\right) $
\begin{equation}
(\sqrt{T}\!-\!\sqrt{n}\!-\!r)^{2} \sigma_{\upsilon}^2 I \!\preceq\! \Upsilon_{T}^{-} (\Upsilon_{T}^{-})^\mathsf{T} \! \preceq \! (\sqrt{T} \!+\! \sqrt{n} \!+\! r)^{2} \sigma_{\upsilon}^2 I. 
\end{equation}
Let $r=\sqrt{2\log{\frac{2}{\delta}}}$ and focus on the right side of $\Upsilon_{T}^{-} (\Upsilon_{T}^{-})^\mathsf{T}$. 
When $T\ge T_{\delta}=(\sqrt{n}+\sqrt{2\log{\frac{2}{\delta}}})^2 /(\sqrt{5}/2-1)^2$, one has with probability $1- \delta$ that 
\begin{align}\label{eq:part4bound}
\frac{3 T \sigma_{\upsilon}^2 }{4} I \preceq \Upsilon_{T}^{-} (\Upsilon_{T}^{-})^\mathsf{T}  \preceq \frac{5T\sigma_{\upsilon}^2}{4} I. 
\end{align}
It follows from (\ref{eq:part4bound}) and $\| \Sigma_{T}^{-1}\|\le1/\lambda_{\min}(V_{dn})$ that 
\begin{align}
\!\!\! \| \Upsilon_{T}^{-} (\Upsilon_{T}^{-})^\mathsf{T} \Sigma_{T}^{-1} \| \!\le\! \| \Upsilon_{T}^{-} (\Upsilon_{T}^{-})^\mathsf{T}\| \|\Sigma_{T}^{-1}\| \!\le\! \frac{5T\sigma_{\upsilon}^2}{4 \lambda_{\min}(V_{dn})}. 
\end{align}

Finally, combining the upper bounds of four parts leads to 
\begin{small}
\begin{equation}
\| E_W \| \le \frac{ 12 \sqrt{n \log \left(\frac{5 \operatorname{det}\left(  V_{up} V_{dn}^{-1}+ I  \right)^{ \frac{1}{2n}}  } {\delta^{1 / n}}\right)}  \! + \! \frac{5T\sigma_{\upsilon}}{4 \sqrt{ \lambda_{\min}(V_{dn}) }} }{\sqrt{\lambda_{\min}(V_{dn})}} ,
\end{equation}
\end{small}
which completes the proof. 
\end{proof}

\subsection{Proof of Theorem \ref{th:converge-speed}}\label{pr:converge-speed}
\begin{proof}
The key point of this proof is to provide a group of explicit $V_{dn}$ and $V_{up}$ about $T$. 
First, we focus on analyzing the case of $\widehat{W}_o$ and then easily extend the analysis to that of $\widehat{W}_c$. 
Note that we are not interested in finding the best $V_{dn}$ and $V_{up}$ rather illustrate their existences. 

First, it has been proved that (in Proposition 8.5. \cite{sarkar2019near}) there exists some scalar functions $\alpha_n$ that depends only on $n$, such that $\Sigma_{T} \succeq \alpha_n T I$. 
Thus, an appropriate group of $\alpha_n$ and $V_{dn}$ can be found such that $ \alpha_n T I \!\preceq\! V_{dn} \!\preceq\! \Sigma_{T}$, which yields that 
\begin{equation}\label{eq:down_bound}
\| V_{dn}^{-1} \| \le \frac{1}{ \lambda_{\min}(\alpha_n TI) } =\frac{1}{ \alpha_n T } \sim \bm{O}(\frac{1}{T}). 
\end{equation}

For $V_{up}$, note that the following inequity always holds and one can use it to determine $V_{up}$, given by 
\begin{equation}
\Sigma_{T} \preceq \operatorname{tr}\left(\sum_{t=0}^{T-1} y_{t} y_{t}^\mathsf{T} \right) I = V_{up}. 
\end{equation}
When $W \!\in\! \mathcal{S}_m$, part of the upper bound of $[\Sigma_{T}]^{ii}$ is given by (\ref{eq:upup}) in the proof of Theorem \ref{th:equivalent-nonstable}, demonstrating $ [\sum_{t=0}^{T-1} y_{t} y_{t}^\mathsf{T} ]^{[ii]} \!\sim\! \bm{O}(T^2)$. 
When $W \!\in\! \mathcal{S}_a$, it is also proved in Theorem \ref{th:equivalent1} that $\| \Sigma_T / T \| \!\stackrel{T\to\infty}{\longrightarrow}\! \| \sigma_{\theta}^2 \sum\limits_{t = 0}^{\infty} W^t (W^t)^\mathsf{T} \!+\! \sigma_{\upsilon}^2 I \|$, which is strictly bounded. 
In element-wise view, since $[V_{up}]^{ii} \le n \max\{[\Sigma_T]^{jj},j\in\mathcal{V}\}$ ($n$ is constant node number), one can directly use the increment property of $\Sigma_T$ to characterize the increment of $V_{up}$. 
Therefore, one can easily infer that
\begin{equation}\label{eq:bound_up}
\|V_{up}\| \sim\left \{
\begin{aligned}
& \bm{O}(T^2) ,~&& \text{if}~W \in \mathcal{S}_m, \\
& \bm{O}(T),~&&\text{if}~W \in \mathcal{S}_a. 
\end{aligned}\right.
\end{equation}
Combine the two factors (\ref{eq:down_bound}) and (\ref{eq:bound_up}), and it follows that 
\begin{align} \label{eq:det}
\operatorname{det} (  V_{up} V_{dn}^{-1}+ I)\sim\left \{
\begin{aligned}
&\bm{O}(T),~&& \text{if}~W \in \mathcal{S}_m, \\
&\bm{O}(1),~&&\text{if}~W \in \mathcal{S}_a. 
\end{aligned} \right.
\end{align}
Taking (\ref{eq:down_bound}) and (\ref{eq:det}) into $(12 \sqrt{n \log(\frac{5 \operatorname{det} (  V_{up} V_{dn}^{-1}+ I)^{ \frac{1}{2n}}  } {\delta^{1 / n}} )}   + \frac{5T\sigma_{\upsilon}}{4 \sqrt{ \lambda_{\min}(V_{dn}) }} )/{\sqrt{\lambda_{\min}(V_{dn})}} $ yields the relationship given by (\ref{eq:wo_bound}), which completes the proof of the first statement. 

Next, consider the the non-asymptotic bound of $\| \widehat{W}_c-W \|$. 
Let $\Sigma_{T,\sigma_{\upsilon} }= T \Sigma_0(T-1)- T \sigma_{\upsilon}^2 W$, and the causality-based estimator (\ref{causality_estimator}) is rewritten as 
\begin{equation}
\widehat{W}_c\!=\!Y_T^+ (Y_T^-)^\mathsf{T} \Sigma_{T,\sigma_{\upsilon} }^{-1}.
\end{equation}
Then, the inference error is given by 
\begin{align}
\widehat{W}_c -W =& (\Theta_{T}+\Upsilon_{T}^{+} - W\Upsilon_{T}^{-}) (Y_{T}^{-})^\mathsf{T} \Sigma_{T,\sigma_{\upsilon} }^{-1} + T\sigma_{\upsilon}^2 W \Sigma_{T,\sigma_{\upsilon} }^{-1} \nonumber \\
=& (\Theta_{T}+\Upsilon_{T}^{+} ) (Y_{T}^{-})^\mathsf{T} \Sigma_{T,\sigma_{\upsilon} }^{-1} - W\Upsilon_{T}^{-} (X_{T}^{-})^\mathsf{T} \Sigma_{T,\sigma_{\upsilon} }^{-1}  \nonumber \\
&+ W(T\sigma_{\upsilon}^2 I- \Upsilon_{T}^{-}(\Upsilon_{T}^{-})^\mathsf{T})\Sigma_{T,\sigma_{\upsilon} }^{-1}. 
\end{align}
Consequently, the upper bound is given by 
\begin{small}
\begin{align}\label{eq:new_bound}
\!\! \| \widehat{W}_c -W  \| & \!\le\! \| \Theta_{T} (Y_{T}^{-})^\mathsf{T} \Sigma_{T,\sigma_{\upsilon} }^{-1} \| + \|\Upsilon_{T}^{+} (Y_{T}^{-})^\mathsf{T} \Sigma_{T,\sigma_{\upsilon} }^{-1}\|  \nonumber \\
 \!\!+\! &\| \Upsilon_{T}^{-} (X_{T}^{-})^\mathsf{T} \Sigma_{T,\sigma_{\upsilon} }^{-1} \| \!+\! \| (T\sigma_{\upsilon}^2 I- \Upsilon_{T}^{-}(\Upsilon_{T}^{-})^\mathsf{T})\Sigma_{T,\sigma_{\upsilon} }^{-1} \|.
\end{align}
\end{small}
\!\!\!\!
Note that the first three terms in (\ref{eq:new_bound}) share the same upper bound forms as $\| \Theta_{T} (Y_{T}^{-})^\mathsf{T} \Sigma_{T}^{-1} \| $, $ \|\Upsilon_{T}^{+} (Y_{T}^{-})^\mathsf{T} \Sigma_{T}^{-1}\|$ and $\| \Upsilon_{T}^{-} (X_{T}^{-})^\mathsf{T} \Sigma_{T}^{-1} \|$ in $\| \widehat{W}_o -W  \|$. 
The proof is similar to that of Theorem \ref{th:asymptotic-performance} and is omitted here. 
As for the last term in RHS of (\ref{eq:new_bound}), one has with high probability that 
\begin{equation}
\mathop {\lim } \limits_{T\to \infty }  \| T \sigma_{\upsilon}^2 I- \Upsilon_{T}^{-}(\Upsilon_{T}^{-})^\mathsf{T}\|/T =0. 
\end{equation}
Therefore, the upper bound of $\| \widehat{W}_c -W  \|$ is determined by the first three terms in (\ref{eq:new_bound}), which converge to zero as $T\to\infty$. 
The second statement in Theorem \ref{th:converge-speed} is proved. 
\end{proof}

\subsection{Proof of Corollary \ref{coro:recursive}}\label{pr:recursive}
\begin{proof}
The whole procedure is similar to the derivation of recursive OLS estimators. 
To begin with, it follows from the analytical expression of $\widehat{W}_{i,c}(t)$ that
\begin{small}
\begin{align}
&\widehat{W}_{i,c}(t)=( \tilde{Z}_{t}^\mathsf{T} \tilde{Z}_{t}^\mathsf{T} -\sigma_{\upsilon}^2 t I )^{-1} \tilde{Z}_{t}^\mathsf{T} \bm{\tilde{b}}_{t}^{i} \nonumber \\
&=\left(\begin{bmatrix} \tilde{Z}_{t-1} \\ \tilde{z}_{t}\end{bmatrix}^\mathsf{T} 
\begin{bmatrix}  \tilde{Z}_{t-1} \\ \tilde{z}_{t}\end{bmatrix}  -\sigma_{\upsilon}^2 tI \right)^{-1} 
\begin{bmatrix} \tilde{Z}_{t-1} \\ \tilde{z}_{t}\end{bmatrix}^\mathsf{T}
\begin{bmatrix} \bm{\tilde{b}}_{t-1}^{i} \\ \tilde{b}_{t}^{i} \end{bmatrix} \nonumber \\
&= \left( \tilde{Z}_{t-1}^\mathsf{T} \tilde{Z}_{t-1} + \tilde{z}_{t}^\mathsf{T} \tilde{z}_{t} -\sigma_{\upsilon}^2 tI \right)^{-1} \left( \tilde{Z}_{t-1}^\mathsf{T}\bm{\tilde{b}}_{t-1}^{i} + \tilde{z}_{t}^\mathsf{T} \tilde{b}_{t}^{i} \right) . 
\end{align}
\end{small}

Next, let $P_t=(\tilde{Z}_{t}^\mathsf{T} \tilde{Z}_{t}^\mathsf{T} -\sigma_{\upsilon}^2 t I)^{-1}$ and it follows that
\begin{align}
&P_t=\left( P_{t-1}^{-1} \!+\! \tilde{z}_{t}^\mathsf{T} \tilde{z}_{t} \!-\!\sigma_{\upsilon}^2 I \right)^{-1},\\
&P_t^{-1} \widehat{W}_{i,c}(t)=\tilde{Z}_{t}^\mathsf{T} \bm{\tilde{b}}_{t}^{i}\label{eq:p_inverse}.
\end{align}
Then, substituting \eqref{eq:p_inverse} into $\widehat{W}_{i,c}(t)$, it yields that 
\begin{align}
&\widehat{W}_{i,c}^{\mathsf{T}}(t) = P_t \left( \tilde{Z}_{t-1}^\mathsf{T}\bm{\tilde{b}}_{t-1}^{i} + \tilde{z}_{t}^\mathsf{T} \tilde{b}_{t}^{i} \right) \nonumber \\
=& P_t \left(P_{t-1}^{-1} \widehat{W}_{i,c}^{\mathsf{T}}(t-1) + \tilde{z}_{t}^\mathsf{T} \tilde{b}_{t}^{i} \right) \nonumber \\
=& P_t \left( \left(P_{t}^{-1} - \tilde{z}_{t}^\mathsf{T} \tilde{z}_{t}  + \sigma_{\upsilon}^2 I  \right) \widehat{W}_{i,c}^{\mathsf{T}}(t-1) + \tilde{z}_{t}^\mathsf{T} \tilde{b}_{t}^{i} \right) \nonumber \\
=& \widehat{W}_{i,c}^{\mathsf{T}}(t-1)  +  P_t \left(  \sigma_{\upsilon}^2 I  - \tilde{z}_{t}^\mathsf{T} \tilde{z}_{t} \right) \widehat{W}_{i,c}^{\mathsf{T}}(t-1)  +  P_t  \tilde{z}_{t}^\mathsf{T} \tilde{b}_{t}^{i} \nonumber \\
=&(I \!+\! \sigma_{\upsilon}^2 P_t) \widehat{W}_{i,c}^{\mathsf{T}}(t-1) \!+\! P_{t} \tilde{z}_{t}^{\mathsf{T}}\left(\tilde{b}_{t}^{i} \!-\!\tilde{z}_{t} \widehat{W}_{i,c}^{\mathsf{T}}(t-1)\right). 
\end{align}
Notice that the invertibility of the symmetric matrix and $(\tilde{z}_{t}^\mathsf{T} \tilde{z}_{t} -\sigma_{\upsilon}^2 I)$ is guaranteed by $\sigma_{\theta}^2\ge\sigma_{\upsilon}^2$, i.e., 
$$\operatorname{det}(\tilde{z}_{t}^\mathsf{T} \tilde{z}_{t} -\sigma_{\upsilon}^2 I)= (\operatorname{tr}(\tilde{z}_{t}^\mathsf{T} \tilde{z}_{t} )-\sigma_{\upsilon}^2)\prod \nolimits_{i=1}^{n-1}(-\sigma_{\upsilon}^2)\neq 0$$ 
almost surely. 
Therefore, one can substitute the eigenvalue decomposition $(\tilde{z}_{t}^\mathsf{T} \tilde{z}_{t} -\sigma_{\upsilon}^2 I)=U_t \Lambda_t U_t^\mathsf{T}$ into $P_t$, and apply the Woodbury formula to obtain $P_{t}$ recursively, given by
\begin{align}\label{eq:recursive_p}
P_{t}=&\left( P_{t-1}^{-1}+U_t \Lambda_t U_t^\mathsf{T} \right)^{-1} \nonumber \\
=&P_{t-1}-P_{t-1} {U_t} \left(\Lambda_t^{-1}+ {U_t^\mathsf{T}} P_{t-1} {U_t} \right)^{-1} {U_t^\mathsf{T}} P_{t-1}. 
\end{align}
The proof is completed. 
\end{proof}


\begin{IEEEbiographynophoto}{Yushan Li}
(S'19) received the B.E. degree in School of Artificial Intelligence and Automation from Huazhong University of Science and Technology, Wuhan, China, in 2018. 
He is currently working toward the Ph.D. degree with the Department of Automation, Shanghai Jiaotong University, Shanghai, China. 
He is a member of Intelligent of Wireless Networking and Cooperative Control group. 
His research interests include robotics, security of cyber-physical system, and distributed computation in multi-agent networks. 
\end{IEEEbiographynophoto}

\begin{IEEEbiographynophoto}{Jianping He} 
(SM'19) is currently an associate professor in the Department of Automation at Shanghai Jiao Tong University. He received the Ph.D. degree in control science and engineering from Zhejiang University, Hangzhou, China, in 2013, and had been a research fellow in the Department of Electrical and Computer Engineering at University of Victoria, Canada, from Dec. 2013 to Mar. 2017. His research interests mainly include the distributed learning, control and optimization, security and privacy in network systems.

Dr. He serves as an Associate Editor for IEEE Trans. Control of Network Systems, IEEE Open Journal of Vehicular Technology, and KSII Trans. Internet and Information Systems. He was also a Guest Editor of IEEE TAC, International Journal of Robust and Nonlinear Control, etc. He was the winner of Outstanding Thesis Award, Chinese Association of Automation, 2015. He received the best paper award from IEEE WCSP'17, the best conference paper award from IEEE PESGM'17, and was a finalist for the best student paper award from IEEE ICCA'17, and the finalist best conference paper award from IEEE VTC'20-FALL.
\end{IEEEbiographynophoto}

\begin{IEEEbiographynophoto}{Cailian Chen}
(M'06) received the B.E. and M.E. degrees in Automatic Control from Yanshan University, P. R. China in 2000 and 2002, respectively, and the Ph.D. degree in Control and Systems from City University of Hong Kong, Hong Kong SAR in 2006. She joined Department of Automation, Shanghai Jiao Tong University in 2008 as an Associate Professor. She is now a Full Professor. 
Before that, she was a postdoctoral research associate in University of Manchester, U.K. (2006-2008). 
She was a Visiting Professor in University of Waterloo, Canada (2013-2014). 
Prof. Chen's research interests include industrial wireless networks, computational intelligence and situation awareness, Internet of Vehicles. 

Prof. Chen has authored 3 research monographs and over 100 referred international journal papers. She is the inventor of more than 20 patents. 
She received the prestigious ”IEEE Transactions on Fuzzy Systems Outstanding Paper Award” in 2008, and Best Paper Award of WCSP17 and YAC18. 
She won the Second Prize of National Natural Science Award from the State Council of China in 2018, First Prize of Natural Science Award from The Ministry of Education of China in 2006 and 2016, respectively, and First Prize of Technological Invention of Shanghai Municipal, China in 2017. She was honored Changjiang Young Scholar in 2015 and Excellent Young Researcher by NSF of China in 2016. 
Prof. Chen has been actively involved in various professional services. 
She serves as Associate Editor of IEEE Transactions on Vehicular Technology, Peerto-peer Networking and Applications (Springer). 
She also served as Guest Editor of IEEE Transactions on Vehicular Technology, TPC Chair of ISAS19, Symposium TPC Co-chair of IEEE Globecom 2016 and VTC2016-fall, Workshop Co-chair of WiOpt18.
\end{IEEEbiographynophoto}

\begin{IEEEbiographynophoto}{Xinping Guan}
(F'18) received the B.S. degree in Mathematics from Harbin Normal University, Harbin, China, in 1986, and the Ph.D. degree in Control Science and Engineering from Harbin Institute of Technology, Harbin, China, in 1999. He is currently a Chair Professor with Shanghai Jiao Tong University, Shanghai, China, where he is the Dean of School of Electronic, Information and Electrical Engineering, and the Director of the Key Laboratory of Systems Control and Information Processing, Ministry of Education of China. 
Before that, he was the Professor and Dean of Electrical Engineering, Yanshan University, Qinhuangdao, China. 

Dr. Guan's current research interests include industrial cyber-physical systems, wireless networking and applications in smart factory, and underwater networks. He has authored and/or coauthored 5 research monographs, more than 270 papers in IEEE Transactions and other peer-reviewed journals, and numerous conference papers. 
As a Principal Investigator, he has finished/been working on many national key projects. He is the leader of the prestigious Innovative Research Team of the National Natural Science Foundation of China (NSFC). 
Dr. Guan is an Executive Committee Member of Chinese Automation Association Council and the Chinese Artificial Intelligence Association Council. 
Dr. Guan received the First Prize of Natural Science Award from the Ministry of Education of China in both 2006 and 2016, and the Second Prize of the National Natural Science Award of China in both 2008 and 2018. 
He was a recipient of IEEE Transactions on Fuzzy Systems Outstanding Paper Award in 2008. He is a National Outstanding Youth honored by NSF of China, Changjiang Scholar by the Ministry of Education of China and State-level Scholar of New Century Bai Qianwan Talent Program of China.
\end{IEEEbiographynophoto}

\end{document}